\documentclass[reqno]{amsart}
\usepackage{mathrsfs, csquotes}
\usepackage{amsfonts,amssymb,amsmath,mathtools,commath,braket}

\usepackage[normalem]{ulem}

\usepackage{color} 
\definecolor{bleu_sombre}{rgb}{0,0,0.6} 
\definecolor{bs}{rgb}{0,0,0.6}  
\definecolor{rouge_sombre}{rgb}{0.8,0,0}
\definecolor{rs}{rgb}{0.8,0,0}
\definecolor{vert_sombre}{rgb}{0,0.6,0}
\definecolor{vs}{rgb}{0,0.6,0}
\usepackage[plainpages=false,colorlinks,linkcolor=bleu_sombre,
citecolor=rouge_sombre,urlcolor=vert_sombre,breaklinks]{hyperref}

\usepackage{enumerate}

\usepackage{a4wide}  

\theoremstyle{plain} 
\newtheorem{theorem}{Theorem}[section]
\newtheorem{lemme}[theorem]{Lemma}
\newtheorem{coro}[theorem]{Corollary}
\newtheorem{prop}[theorem]{Proposition}
\newtheorem{assumption}{Assumption}
\newtheorem{property}[theorem]{Property}

\theoremstyle{definition}
\newtheorem{remark}[theorem]{Remark}

\newtheorem{exemple}[theorem]{Example}
\newtheorem{defi}[theorem]{Definition}

\newcommand{\iu}{{\rm i}}

\def\ec{{\mathbb E}}
\def\ecL{{\mathbb E}} 

\def\pc{{\mathbb P}}
\def\CC{{\mathbb C}}
\def\RR{{\mathbb R}}
\def\NN{{\mathbb N}}
\def\ZZ{{\mathbb Z}}

\def\Hilb{{\mathcal H}}
\def\Dom{{\mathcal D}}

\def\mass{\mu}

\def\vt{\vartheta}

\def\({\left(}
\def\){\right)}
\def\<{\left\langle}
\def\>{\right\rangle}

\numberwithin{equation}{section}

\newcommand{\supp}{\mathrm{supp}}

\newcommand{\dist}{\mathrm{dist}}

\newcommand{\tr}{\mathrm{tr}}

\newcommand{\be}{\begin{equation}}
\newcommand{\ee}{\end{equation}}
\newcommand{\bea}{\begin{eqnarray}}
\newcommand{\eea}{\end{eqnarray}}
\newcommand{\bee}{\begin{eqnarray*}}
\newcommand{\eee}{\end{eqnarray*}}

\begin{document}

\title[Localization properties of GAL's]{Localization for gapped Dirac hamiltonians with random perturbations: application to Graphene Antidot Lattices}
\author[J.-M. Barbaroux]{J.-M. Barbaroux}
\email{jean-marie.barbaroux@univ-tln.fr}
\address[J.-M. Barbaroux]{Aix Marseille Univ, Universit\'e de Toulon, CNRS, CPT, Marseille, France}
\author[H.D. Cornean]{H.D. Cornean}
\email{cornean@math.aau.dk}
\address[H.D. Cornean]{Department of Mathematical Sciences, Aalborg University\\
Fredrik Bajers Vej 7G, 9220 Aalborg \O, Denmark}
\author[S. Zalczer]{S. Zalczer}
\email{sylvain.zalczer@univ-tln.fr}
\address[S. Zalczer]{Aix Marseille Univ, Universit\'e de Toulon, CNRS, CPT, Marseille, France}

\begin{abstract}
In this paper we study random perturbations of first order elliptic operators with periodic potentials. 
We are mostly interested in Hamiltonians modeling graphene antidot lattices with impurities. The unperturbed operator $H_0 := D_S + V_0$ is the sum of a Dirac-like operator $D_S$ plus a potential $V_0$, and is assumed to have an open gap. The random potential $V_{\omega}$ is of  Anderson-type with independent, identically distributed coupling constants and moving centers, with absolutely continuous probability distributions. We prove band edge localization, namely that there exists an interval of energies in the unperturbed gap where the almost sure spectrum of the family $H_{\omega} := H_0 +V_{\omega}$ is dense pure point, with exponentially decaying eigenfunctions, that give rise to dynamical localization.
\end{abstract}
\maketitle
\setcounter{tocdepth}{1}
\tableofcontents

\section{Introduction}
The main goal of this paper is to derive spectral and dynamical localization properties near band edges for first order elliptic and periodic operators densely defined in $L^2(\RR^d, \CC^n)$, perturbed by random potentials. The main application we have in mind is related to graphene antidot lattices. Graphene is a two-dimensional material made of carbon atoms arranged in a honeycomb structure  
\cite{CGPNGG}. 
Charge carriers close to the Fermi energy behave like massless Dirac fermions, making pristine graphene a semimetal.  
One needs to produce an energy gap in order to turn  graphene into a semiconductor. 

Several gapped models have been proposed in the physics literature  (see \cite{DOW} and references therein). One such setting consists of a graphene sheet with periodic nanoscale perforations, a so-called graphene antidot lattice (GAL). Here the quantum dynamics is given by a two dimensional Dirac operator with a periodic mass term. Under certain conditions, a spectral gap appears near the zero energy (see \cite{BTP} and \cite{FPFMBPJ} for theoretical works and \cite{barbaroux} for a mathematical study).

The next step is to perturb the gapped Hamiltonian by an Anderson type potential  for modeling sample impurities. 
There are two types of properties of such Hamiltonians we are interested in (see Definition~\ref{def:all-loc} for details):
\begin{itemize}
\item \emph{Spectral localization}: Dense pure point spectrum near the Fermi level with  exponentially decaying associated eigenfunctions. 
\item \emph{Dynamical localization}: Uniform boundedness in time of moments of positive orders of states which are spectrally  supported in the dense point spectrum. 
\end{itemize}
Starting from the seminal  contributions by Anderson \cite{Anderson} and the rigorous spectral analysis initiated by Pastur \cite{Pastur, GMP77}, a significant number of papers on Anderson-like Hamiltonians have been published in the mathematical literature.

Most of the existing mathematical results regarding these properties are derived for the case were the kinetic energy  is described by discrete or continuous Laplace operators. The case where the kinetic energy is given by Dirac or Maxwell operators has been the subject of studies only recently. 

A step towards Dirac operators has been done in the case where the kinetic energy is given by a Laplacian on $L^2(\RR^d)\otimes \CC^\nu$, $\nu>1$ and the random potential is matrix valued (see \cite{BN2015} and references therein). 
In \cite{PO2005, POC2017} the authors considered discretized versions of Dirac operators on $\ell^2(\ZZ^d, \CC^\nu)$ ($d=1,2,3$), with a simple mass potential, and a random potential given by a matrix valued diagonal operator, and proved spectral and dynamical localization near band edges. 

A precise analysis of the conditions leading to localization enables us to provide a result not only for 2-dimensional continuous Dirac operators, but also for a larger class of first order elliptic operators.  This includes the operators describing ``classical waves" as defined by Klein and Koines \cite{KleinKoines2004}.

In our paper we are mainly interested in the case in which a spectral gap is created near the Fermi level by a deterministic multiplicative potential, which afterwards is perturbed by a random one.

Our main results on spectral and dynamical localization are stated in Theorem~\ref{thm:spec-loc} and Theorem~\ref{thm:dyn-loc}. The proofs of these results exploit the developments of the theory of multi-scale analysis for continuous operators as given by \cite{CH, GDB, GKbootstrap, klein}.
\section{Setting and main results}

We start with a few definitions. 

\begin{defi}

Let $\{\sigma_i\}_{i=1}^d$ be a family of $n\times n$ Hermitian matrices where 
$n,d \geqslant 1$. 
We consider the following \emph{first-order linear operator} with constant coefficients: 
\begin{equation}\label{def:D0}                                     
 \sigma\cdot(-\iu\nabla) := \sum_{j=1}^d \sigma_j(-\iu\frac{\partial}{\partial x_j}),
\end{equation}
densely defined in $L^2(\RR^d, \CC^n)$. It is {\em elliptic} if there exists $C>0$ such that for all $p\in\RR^d$ and $q\in\CC^n$ we have 

\begin{align}\label{hc2}
\|(\sigma\cdot p) q\|_{\CC^n} \geqslant C\|p\|_{\RR^d}\;\|q\|_{\CC^n}.
\end{align}
\end{defi}
%
%
%
If $E_0\in \RR$, the maps
 $$
 	\RR^d\ni p\mapsto g_{ij}(p)
 	:= \left [(\sigma\cdot p-E_0-\iu)^{-1}\right ]_{ij}\in \CC,
 	\quad 1 \leqslant i,j \leqslant n,
 $$ 
are well defined and due to \eqref{hc2} there exists a constant $C<\infty$ such that 
\begin{align}\label{hc3}
|g_{ij}(p)| \leqslant C \langle p\rangle^{-1},\quad 1 \leqslant i, j \leqslant n
\end{align}
where $\langle p\rangle := \sqrt{1+|p|^2}$ for  some norm $|\cdot|$ on $\RR^d$.

A direct consequence is that $\sigma\cdot (-\iu \nabla)$ is self-adjoint on the Sobolev space $H^1(\RR^d,\CC^n)$.  
\color{black}
\begin{defi}
 We say that an operator on $L^2(\RR^d, \CC^n)$ is a {\em   coefficient positive operator} if it is a bounded invertible operator given by the multiplication by an $n\times n$ Hermitian matrix-valued measurable function
 $S(x)$ such that there exist two positive constants $S_\pm$ such that:
 \begin{equation}\label{def:S}
  0<S_-I_n \leqslant S(x) \leqslant S_+I_n,
 \end{equation}
where $I_n$ is the $n\times n$ identity matrix. 
\end{defi}

We consider operators of the type
\begin{equation}\label{hc1}
H_0=S D_0 S+V_0
\end{equation}
where $D_0$ is a first-order elliptic operator with constant coefficients like in \eqref{def:D0}, and $S$ is a coefficient positive operator as in  \eqref{def:S}.
The function $S\in W^{1,\infty}(\RR^d, \mathcal{H}_n(\CC))$, where $\mathcal{H}_n$ is the space of $n\times n$ Hermitian matrices, is supposed to be $\ZZ^d$-periodic. We denote 
 $$
 D_S := S D_0 S.
 $$
Such operators appear in connection with wave propagation and are sometimes called  \emph{classical wave operators} (cf. \cite{KKS, KleinKoines2004}). We warn the reader that this name has nothing to do with the M\"oller wave operators of quantum scattering theory. The potential $V_0$ is $\ZZ^d$-periodic and belongs to $ L^\infty(\mathbb{R}^d, \mathcal{H}_n)$.  

With the above definitions and assumptions the operator $H_0$ is self-adjoint on $H^1(\RR^d, \CC^n)$.

\begin{assumption}[gap assumption]\label{assump-gap} The spectrum of $H_0$ contains a finite open gap,  which will be denoted $(B_-, B_+)$.
\end{assumption}

\begin{exemple}\label{ex1}
The simplest examples are the free Dirac operators with mass $\mass>0$ in dimension two and three, respectively given by
$$
   H_0 = \sigma_1 (-\iu\partial_{x_1}) + \sigma_2 (-\iu\partial_{x_2}) + \mass\sigma_3
   \mbox{ in } L^2(\RR^2,\CC^2),
   \mbox{with $\sigma_i$ being the Pauli matrices, } 
$$
\[\sigma_1=\begin{pmatrix}0&1\\1&0\end{pmatrix},\quad \sigma_2=\begin{pmatrix}0&\iu\\-\iu&0\end{pmatrix} ,\quad \sigma_3=\begin{pmatrix}1&0\\-0&-1\end{pmatrix} \]
and
$$
   H_0 = \boldsymbol{\alpha}\cdot(-\iu\nabla) +  \mass {\beta}
   \mbox{ in } L^2(\RR^3,\CC^4),\mbox{ with }
   \boldsymbol{\alpha} =(\alpha_1, \alpha_2, \alpha_3), \mbox{ $\boldsymbol{\beta}$ being the Dirac matrices}
$$
\[\alpha_i=\begin{pmatrix}0&\sigma_i\\\sigma_i&0\end{pmatrix},\quad \beta=\begin{pmatrix}{\bf1}&0\\0&-{\bf1}\end{pmatrix}. \]
Both operators are such that $\rho(H_0)\cap\RR = (-\mass,\, \mass)$ (cf. \cite{thaller}), where for $T$ self-adjoint, $\rho(T)$ is its resolvent set. 
\end{exemple}

\begin{exemple}\label{ex2}
A family of operators which is physically relevant in connection to graphene antidot lattices, as introduced e.g. 
in \cite{Pedersen2} and rigorously studied in \cite{barbaroux}, is the following:
\begin{equation*}
 H_0(\alpha,\beta)=D_0+\beta \sum_{\gamma \in \ZZ^2}\chi\left(\frac{\cdot -\gamma }{\alpha }\right)\sigma_3\,  \quad \mbox{in $L^ 2(\RR^ 2,\CC^2)$}, 
\end{equation*}
where $D_0=\sigma\cdot (-\iu \nabla)$ is the two-dimensional massless Dirac operator, $\beta>0$, $\alpha\in (0,1]$, and 
$\chi: \RR^2\to\RR$ is a bounded function with support in a compact subset 
of $(-\frac{1}{2},\frac{1}{2}]^2$. 

If $\int \chi\neq 0$ it has been proved  in \cite[Theorem~1.1]{barbaroux} the existence of a spectral gap near zero for this operator, namely that there exist constants $C,C'>0$ and $\delta \in (0,1)$ such that for every $\alpha \in  (0,1/2]$ and $\beta>0$ satisfying $\alpha\beta<\min\{\delta,C'/C\}$  we have
 $$
 [-\alpha^2\beta(C'-C\alpha\beta),\, \alpha^2\beta(C'-C\alpha\beta)]
 \subset \rho(H_0(\alpha,\beta)).
 $$
\end{exemple}
\begin{exemple}\label{ex3}
In \cite{FK} it has been shown that certain operators of the type $D_S$ as in \eqref{hc1}, modeling Maxwell operators with periodic dielectric constants, can also   have open gaps.  
\end{exemple}
\color{black}
\medskip

For operators fulfilling Assumption~\ref{assump-gap}, we want to study the effect of random perturbations on the  spectral gap $(B_-, B_+)$.  

The random matrix-valued perturbation $V_{\omega}$ describing local defects is defined by
\begin{equation*}
 V_{\omega} =
 \sum_{i\in\mathbb{Z}^d}\lambda_{i}(\omega)u(\cdot-\xi_{i}(\omega)-i) ,
\end{equation*}
for some $u$, $\lambda_i$ and $\xi_i$ satisfying Assumption~\ref{assump:2} below.
The total Hamiltonian is thus
\begin{equation}\label{defop}
 H_{\omega} = H_0 + V_{\omega} .
\end{equation}
 
\begin{assumption}\label{assump:2}
  \noindent (i) The real-valued random variables 
  $\{\lambda_{i}(\omega),i\in\mathbb{Z}^d\}$ are independent and identically distributed. Their common distribution is absolutely continuous with respect to Lebesgue measure, with a density $h$ such that $\|h\|_{L^\infty}<\infty$. 
We assume that $\supp (h) = [-m,M] \neq \{0\}$ for some finite non-negative $m$ and $M$. 
  
 \noindent (ii) The variables $\{\xi_i(\omega),\, i\in\mathbb{Z}^d\}$ are independent and identically distributed, and they are also independent from the $\lambda_j$'s.  They take values in $B_R$ with $0<R<\frac{1}{2}$, where $B_R$ is the ball in $\RR^d$ with radius $R$ and centered at the origin.

 \noindent (iii) The single-site matrix potential $u$ is compactly supported with $\mathrm{supp}(u)\subset [-2, 2 ]^d$. In addition, $u$ is assumed to be continuous almost everywhere,  with $u\in L^\infty(\mathbb{R}^d, \mathcal{H}_n^+)$, where $\mathcal{H}_n^+$ is the space of $n\times n$ non-negative Hermitian matrices.

 \noindent (iv) The density $h$ decays sufficiently rapidly near $-m$ and $M$, i.e. 
 $$
 0<\mathbb{P}\left\{|\lambda+m|<\epsilon\right\}\leqslant\epsilon^{d/2+\beta},
 $$
 $$
 0<\mathbb{P}\left\{|\lambda-M|<\epsilon\right\}\leqslant\epsilon^{d/2+\beta}
 $$ 
for some $\beta>0$.
\end{assumption}
%
%
%
\begin{remark}\label{rem:thm}Here are a few comments:

\noindent (i) We take as probability space $\Omega=\left(\supp (h)\right)^{\mathbb{Z}^d}\times(B_R)^{\mathbb{Z}^d}$ equipped with the product probability measure.

\noindent (ii) The periodicity of $V_0$ and $S$, and hypotheses (i) and (ii) imply that the family $\{H_{\omega},\omega\in\Omega\}$ 
 has a deterministic spectrum $\Sigma$ in the sense that there exists $A_0\subset\Omega$ with probability 1 such that $\forall\omega\in A_0,\sigma(H_\omega)=\Sigma$ (cf. for example \cite[Theorem 4.3, p20]{klein}).

\noindent (iii)
A standard result about trace estimates \cite[Theorem 4.1]{simon} states that
$$
f(x)g(-i\nabla)\in\mathcal{T}_q\quad \mbox{if}\quad f,g\in L^q(\RR^d) \quad
\mbox{for}\quad  2\leqslant q < \infty   \ 
$$
with 
\begin{equation*}
  \|f(x)g(-i\nabla)\|_q \leqslant (2\pi)^{- d/q}\|f\|_{L^q}
  \|g\|_{L^q}
\end{equation*}
where $\mathcal{T}_q$ denotes the trace ideal and $\|\cdot\|_q$ the associated norm.

If $q>d$, each $g_{ij}\in L^q(\RR^d)$.  Thus if $f\in L^q(\mathbb{R}^d,\mathcal{M}_n(\CC))$  we obtain that $f(\cdot)(D_0-E_0-\iu)^{-1}\in\mathcal{T}_q$ and there exists a constant $C<\infty$  such that for all $E_0\in\RR$ and $f\in L^q(\mathbb{R}^d,\mathcal{M}_n(\CC))$ one has 
\begin{equation*}
      \|f(\cdot )(D_0-E_0-\iu)^{-1}\|_q \leqslant C
      \max_{1 \leqslant i,j \leqslant n}\|f_{ij}\|_{L^q(\RR^d)}.
     \end{equation*}
     In order to simplify notation we will sometimes forget about the matrix structure of the various objects and simply write for example $\|f\|_{L^q}$ instead of taking the maximum over all its $n^2$ components. 
     
 Denote for simplicity $z=E_0 +i$. We have

$$
  (D_S-z)^{-1}=S^{-1}(D_0-z S^{-2})^{-1}S^{-1}
$$
and 
\begin{equation*}
 (D_0-z S^{-2})^{-1}=(D_0-z)^{-1}- (D_0-z)^{-1}z(I_n- S^{-2} )(D_0-z S^{-2} )^{-1}.
\end{equation*}
A consequence of \eqref{def:S} is that the entries of $S$ and those of $S^{-1}$ are globally bounded.
Hence, for any bounded interval $I\subset\RR$, there  exists a finite constant $C_I$  such that for any $E_0\in I$ and $f\in L^q$ we have:
\begin{equation*}
      \|f(\cdot)(D_S-E_0-i)^{-1}\|_q\leqslant 
      C{_I}\;
      \|f\|_{L^q}.
\end{equation*}
     
If $E_0\in (B_-,\, B_+)$ we have that $(H_0-E_0)^{-1}$ exists as a bounded operator. Then by using both the first resolvent identity to change $E_0$ with $E_0+\iu$ and the second resolvent identity to produce a $(D_S-E_0-\iu)^{-1}$ to the left, we find 
$f(\cdot )(H_0-E_0)^{-1}\in\mathcal{T}_q$ if $q>d$ and that for any compact subinterval $J$ of $(B_-,B_+)$ there  exists a finite constant $C'_J$  such that for any $E_0\in J$ and $f\in L^q$ we have:
\begin{equation}\label{schatten}
\|f(\cdot )(H_0-E_0)^{-1}\|_q \leqslant C{'_J}\|f\|_{L^q} .
\end{equation}

\noindent (iv) Hypotheses (i)-(iii) imply that $\forall \omega$, $\|V_{\omega}\|_\infty\leqslant C$ where $C$ is a finite constant depending only on $m$, $M$, $u$ and $R$.

\noindent (v) As a consequence, the operator $H_\omega$ is self-adjoint on $H^1(\RR^d,\CC^n)$ for any $\omega$.

\noindent (vi) Another useful result is the following. Given a Schwartz function $\chi\in\mathcal{S}(\RR^d,\CC)$, since $S\in W^{1,\infty}(\RR^d, \mathcal{M}_n(\CC))$, the commutator $[H_0,\chi]$ is bounded.
Indeed, we have:
\begin{align*}
  [H_0,\chi]=S\left (\sigma \cdot(-\iu\nabla\chi)\right )S.
\end{align*}

\end{remark}

We denote: 
\begin{equation}\label{Minfini}
 M_\infty := \max\{m,M\} \sup_{ (x_i) \in [-\frac12, \frac12]^{\ZZ^d}}\left\|
 \sum_{i\in\ZZ^d} u(\cdot-x_i-i)\right\|_\infty<\infty ,
 \end{equation}
 
where $\|\cdot\|_\infty$ means the supremum on $\RR^d$ of the operator norm associated with the standard Euclidian norm on $\CC^n$. Remember that $u$ has compact support thus only a finite numbers of terms are different from zero in the above series.

Next, we need an assumption on the almost sure spectrum. In Proposition~\ref{pastout} we will give  sufficient conditions which make sure that it holds.
\begin{assumption}\label{gap}
Let $\Sigma$ be the almost sure spectrum of $H_\omega$. Then there exist two constants $B'_\pm$ satisfying $B_- \leqslant B'_- < B'_+ \leqslant B_+$ such that
$$
 \Sigma\cap\left( (B_-,B'_-)\cup
 (B'_+,B_+)\right)\neq\varnothing \quad \text{and} \quad 
\Sigma\cap(B'_-,B'_+)=\varnothing,$$ i.e. some new almost sure spectrum appears in the old gap, while a smaller gap still exists. 

\end{assumption}

Due to \cite[Theorem~1, \S6, p304]{kirsch1989} we have information on the spectrum not only for almost every $\omega$ but for all $\omega\in\Omega$. 
\color{black}
\begin{defi}
 We say that an ergodic family of operators $(H_\omega)_{\omega\in\Omega}$ is Kirsch-standard if:
 \begin{enumerate}
  \item $\Omega$ is a Polish space and the $\sigma$-algebra contains the Borel sets on $\Omega$.
  
  \item There is a set $\Omega_0$ with probability one such that 
  $H_\omega$ is self-adjoint for any $\omega\in\Omega_0$ 
  and the mapping $\omega\mapsto H_\omega$ restricted to $\Omega_0$ is continuous in the sense that if $\omega_j\to\omega$ then $H_{\omega_j}\to H_\omega$ in the sense of strong resolvent convergence.
 \end{enumerate}

\end{defi}

Let us briefly show that in our case we deal with a Kirsch-standard ergodic family of operators with $\Omega_0=\Omega$. First, $\Omega$ is a Polish  space as a countable product of Polish spaces  when it is equipped with the classical distance on a product of metric spaces. Second, it suffices to show that for any $\phi\in\mathcal{C}^\infty_c(\RR^d,\CC^n)$ we have $H_{\omega_j}\phi\to H_\omega\phi$ when $\omega_j\to\omega$ (cf. \cite[Theorem VIII.25]{RS1}).

If $\omega_j\to\omega\in\Omega$, then for all $i\in\ZZ^d$ $\lambda_i(\omega_j) \to\lambda_i(\omega)$ and $\xi_i(\omega_j) \to\xi_i(\omega)$.
Then (assuming for simplicity $n=1$):  
 $$
 \|H_{\omega_j}\phi-H_\omega\phi\|^2
 =\int_{\RR^d}\sum_{i\in\ZZ^d}
 \left|\lambda_i(\omega_j)u(\cdot-\xi_i(\omega_j)-i)
 -\lambda_i(\omega)u(\cdot-\xi_i(\omega)-i)\right|^2|\phi|^2.
 $$
As $u$ is continuous almost everywhere, the difference in the integral tends almost everywhere to 0 and the integrand is 
bounded by $4M_\infty^2|\phi|^2$ which is integrable. Using the dominated convergence theorem,
we find the desired result.

Note that if $\xi_i(\omega)$ takes only discrete values (including the case where it is constant), we do not need the continuity of $u$.

The fact that $(H_\omega)$ is a standard ergodic family of operators has the important consequence that (see \cite[Theorem~1, \S6, p304]{kirsch1989})
\begin{equation}
 \forall\omega\in\Omega,\ \sigma(H_\omega)\subset \Sigma.\label{spectre}
\end{equation}

Hence $\Sigma$ only depends on the support of the probability distributions. Also, 
$\Sigma \cap [B_-,B_+]$ is characterized by the  following two propositions which state that  under Assumptions~\ref{assump-gap} and \ref{assump:2} one can tune the parameters in such a way that Assumption~\ref{gap} holds and some ``new" almost sure  spectrum appears in the old gap, without closing it though. Moreover, the almost sure spectrum has exactly one (smaller) gap in the given gap of the unperturbed operator. Proofs will be given in Appendix~\ref{specloc}.
%
%
%
\begin{prop}\label{pastout}
There exist $u$, $m$, and $M$ as in Assumption~\ref{assump:2} such that $H_\omega$ satisfies Assumption~\ref{gap}.
\end{prop}
%
%
%
\begin{prop}[Location of the spectrum in the gap of $H_0$]\label{projspec2}
Assume the existence of $B'_-$ and $B'_+$ of Assumption \ref{gap}.  Denote   
\begin{equation*}
    \tilde{B}_-=\sup\{E\in\Sigma\ |\ E\leqslant B'_-\}\quad\mbox{and}\quad 
    \tilde{B}_+=\inf\{E\in\Sigma\ |\ E\geqslant B'_+\}.
\end{equation*} 
Then $[B_-,\tilde{B}_-]\subset \Sigma$ and $[\tilde{B}_+,B_+]\subset \Sigma$.
\end{prop}

Our main results on localization are the following.
\begin{theorem}[Spectral localization]\label{thm:spec-loc}
Under Assumptions~1, 2 and 3, there exist two constants $E_\pm$ satisfying $B_-\leqslant E_- \leqslant B'_-$ and 
$ B'_+ \leqslant E_+\leqslant B_+$ such that $\Sigma\cap (E_-,E_+)$ is non-empty, dense pure point, with exponentially decaying eigenfunctions.
\end{theorem}
%
%
%
\begin{theorem}[Dynamical localization]\label{thm:dyn-loc}
Suppose Assumptions~1, 2 and 3 hold, and denote $E_\pm$  the two energies of Theorem~\ref{thm:spec-loc}. If $r>0$ and $\psi\in L^2(\mathbb{R}^d,\mathbb{C}^n)$ has compact support, then for any compact interval $J\subset(E_-,E_+)$,
\begin{equation}\label{def:dyn-loc}
	\ec\left\{\| \, |x|^r E_\omega(J) e^{-iH_{\omega}t}\psi\|^2\right\}<\infty
\end{equation}
where $E_\omega(J)$ denotes the spectral projector on the interval $J$ for $H_\omega$ and $\ec$ is the expectation associated to $\mathbb{P}$.
\end{theorem}
%
%
%
Throughout this article, we shall use the sup norm in $\RR^d$
\begin{equation}\label{dist}
  |x| = \max\{ |x_i| : i=1,\ldots,d\} .
\end{equation}
%
%
%
\begin{remark}
Some stronger dynamical localization results will be described in the next section, see in particular the estimate \eqref{eq:strong-loc} which will be proved in Theorem~\ref{thm:main-in-the-text}. In particular, Theorem~\ref{thm:dyn-loc} is a straightforward consequence of Theorem~\ref{thm:main-in-the-text}. 
\end{remark}

\section{One method to localize them all: Germinet and Klein's bootstrap multiscale analysis}

Here we briefly explain how Germinet and Klein's multiscale analysis has to be applied in our setting. 
More details can be found in \cite{GKbootstrap} and \cite{klein}. 

In this section, $H_\omega$ denotes an ergodic random self-adjoint operator on $L^2(\RR^d,\CC^n)$. 

\subsection{Spectral and dynamical localization}
Given a set $B\subset\RR^d$, we denote $\chi_B$ its characteristic function. For $x\in\ZZ^d$, we denote $\chi_x$ the characteristic function of the cube of side-length $1$ centered at $x$.
We recall that $\langle x\rangle =\sqrt{1+|x|^2}$. The projection-valued spectral measure of $H_\omega$ will be denoted by $E_\omega(\cdot)$. The Hilbert-Schmidt norm of an operator $A$ is denoted by $\|A\|_2$.
\begin{defi}\label{def:all-loc}
 Let $H_\omega$ be an ergodic random operator defined on a probability space $(\Omega,\mathcal{F},\pc)$ and $I$ an open interval. The different localization properties are the following:
 \begin{enumerate}
  \item  
The family of operators $(H_\omega)$ exhibits exponential localization (EL) in $I$ if it has only pure point spectrum in $I$ and for $\pc$- almost every $\omega$ the eigenfunctions of $H_\omega$ with eigenvalue in $I$ decay exponentially in the $L^2$ sense, i.e. for $\pc$- almost every $\omega$, for any eigenvalue $E$ in $I$ and any associated eigenfunction $\psi_E$, there exist constants $C$ and $m>0$ such that for all $x\in\ZZ^d$, 
$\|\chi_x\psi_E\|\leqslant Ce^{-m|x|}$.

\item  $H_\omega$ exhibits strong dynamical localization (SDL) in $I$ if $\Sigma\cap I\neq\varnothing$ and for  each compact interval $J\subset I$ and $\psi\in\Hilb$ with compact support, we have 
  \begin{equation*}
   \ec\left\{\sup_{t\in\RR}\|\langle x\rangle ^{r}E_\omega(J)e^{-itH_\omega}\psi\|^2\right\}<\infty\text{ for all }r \geqslant 0.
  \end{equation*}
\item $H_\omega$ exhibits strong sub-exponential Hilbert-Schmidt-kernel decay (SSEHSKD) in $I$  if $\Sigma\cap I\neq\varnothing$ and for  each compact interval $J\subset I$ and
$0<\zeta<1$ there is a finite constant $C_{J,\zeta}$ such that 
\begin{equation}\label{eq:strong-loc}
 \ec\left\{\sup_{\|f\|_{{\infty}}\leqslant 1}\|\chi_x 
 E_\omega(J)f(H_\omega)\chi_y\|_2^2\right\}\leqslant C_{I,\zeta}e^{-|x-y|^\zeta},
\end{equation}
for all $x$, $y\in\ZZ^d$, the supremum being taken over all Borel functions $f$ of a real variable, with $\|f\|_{{\infty}}=\sup_{t\in\RR}|f(t)|$, and $\|\cdot\|_2$ is the Hilbert-Schmidt norm.
 \end{enumerate}

\end{defi}

Other types of localization are presented in \cite{klein} but they are all implied by (SSEHSKD). Note that (SDL) is also implied by (SSEHSKD).

As in \cite{klein}, we define $\Sigma_{EL}$ (resp. $\Sigma_{SSEHSKD}$) as the set of $E\in\Sigma$ for which there exists an open interval $I\ni E$ such that $H_\omega$
exhibits exponential localization (resp. strong sub-exponential Hilbert-Schmidt kernel decay) 
 in $I$.

\subsection{Generalized eigenfunction expansion}\label{S:section3.2}

Let $\Hilb=L^2(\RR^d,dx;\CC^n)$. Given $\nu>d/4$, we define the weighted spaces $\Hilb_\pm$ as
 $$
 \Hilb_\pm=L^2(\RR^d,\langle x\rangle ^{\pm4\nu}dx;\CC^n).
 $$

The sesquilinear form
 $$
 \langle\phi_1,\phi_2\rangle_{\Hilb_+,\Hilb_-}=\int\bar{\phi_1}(x)\cdot\phi_2(x)dx
 $$
where $\phi_1\in\Hilb_+$ and $\phi_2\in\Hilb_-$ is the duality map.
\color{black}

We set $T$ to be the self-adjoint operator on $\Hilb$ given by multiplication by the function $\langle x\rangle ^{2\nu}$ ; note that $T^{-1}$ is bounded.

\begin{property}[SGEE]
 We say that an ergodic random operator $H_\omega$ satisfies the strong property of  generalized eigenfunction expansion (SGEE) in some open interval $I$ if, for some $\nu>d/4$, 
 \begin{enumerate}
  \item The set
 $$
 \mathcal{D}^\omega_+=\{\phi\in\mathcal{D}(H_\omega)\cap \mathcal{H}_+; 
 H_\omega\phi\in \Hilb_+\}
 $$
 is dense in $\Hilb_+$ and is an operator core for $H_\omega$ with probability one.
 
 \item There exists a bounded, continuous function $f$ on $\RR$, strictly positive on the spectrum of $H_\omega$
 such that
 \begin{equation*}
  \ec\left\{[ \tr(T^{-1}f(H_\omega) E_\omega(I)T^{-1})]^2\right\}<\infty.
 \end{equation*}

 \end{enumerate}
\end{property}
%
%
%
\begin{defi}
A measurable function $\psi:\RR^d\to\CC^n$ is said to be a generalized eigenfunction of $H_\omega$ with generalized eigenvalue $\lambda$ if $\psi\in\Hilb_-\backslash\{0\}$ 
and 
 $$
 \langle H_\omega\phi,\psi \rangle_{\Hilb_+,\Hilb_-}=\lambda\langle \phi,
 \psi \rangle_{\Hilb_+,\Hilb_-},\ \mbox{ for all } \phi\in\Dom_+^\omega.
 $$
\end{defi}
%
%
%
As explained in \cite{klein}, when (SGEE) holds, a generalized eigenfunction which is in $\Hilb$ is a bona fide eigenfunction. 
Moreover, if $\mu_\omega$ is the spectral measure for the restriction of $H_\omega$ to the Hilbert space $E_\omega(I)\Hilb$, then $\mu_\omega$-almost every $\lambda$ is a generalized eigenvalue of $H_\omega$.

\subsection{Finite volume operators and their properties}

%

We remind the reader that throughout this article we use the sup norm in $\RR^d$ :
$|x|=\max\{|x_i|:i=1, \ldots ,d \}$. By $\Lambda_L(x)$ we denote the open box of side $L>0$ centered at $x\in\RR^d$:
  $$
  \Lambda_L(x)=\{y\in\RR^d;|y-x|<\frac{L}{2}\},
  $$
and by $\bar{\Lambda}_L(x)$ the closed box. We define the boundary belt as 
  $$
  \Upsilon_L(x)= \bar{\Lambda}_{L-1}(x)\backslash \Lambda_{L-3}(x).
  $$

We will write $\Lambda_l\sqsubset\Lambda_L(x)$ when a smaller box $\Lambda_l$ is completely surrounded by the belt $\Upsilon_L(x)$ of a bigger 
 box $\Lambda_L(x)$.  More precisely, this means that if $x\in \ZZ^d$ and $L>l+3$ we have $\Lambda_l\subset\Lambda_{L-3}(x)$. 

Given a box $\Lambda_L(x)$, we define the localized operator
\begin{equation}\label{eq:1.4}
 H_{\omega,x,L}=H_0+\sum_{i\in\Lambda_L(x)\cap\ZZ^d}\lambda_i(\omega)u_i(\cdot -\xi_i(\omega))
 = H_0 + V_{\omega,x,L},
\end{equation}where we denote $u_i=u(\cdot-i)$. This operator is a self-adjoint unbounded operator on $L^2(\RR^d,\CC^n)$.

We can then define $R_{\omega,x,L}(z)=(H_{\omega,x,L}-z)^{-1}$ the resolvent of 
$H_{\omega,x,L}$ and $E_{\omega,x,L}(\cdot)$ its spectral projection.

\begin{defi}
We say that  an ergodic random family of operators $H_\omega$ is Klein-standard \cite{klein} if for each $x\in\ZZ^d$, $L\in\NN$ there is a measurable map $H_{\cdot,x,L}$ from $\Omega$ to self-adjoint operators
 on $L^2(\RR^d,\CC^n)$ such that 
 $$
   U(y)H_{\omega,x,L}U(-y)=H_{\tau_y\omega,x+y,L}
 $$ 
 where $\tau$ and $U$ define the ergodicity:
 $$
 U(y)H_\omega U(y)^*=H_{\tau_y(\omega)}.
 $$
\end{defi}
It is easy to see that the family \eqref{eq:1.4} of localized operators makes $H_\omega$ a Klein standard operator.

We now enumerate the properties which are needed for multiscale analysis to be performed, yielding thus various localization properties.

\begin{defi}
An event is said to be based in a box $\Lambda_L(x)$ if it is determined by 
conditions on the finite volume operators 
$(H_{\omega,x,L})_{\omega\in\Omega}$.
\end{defi}
\begin{property}[IAD]
Events based in disjoint boxes are independent.
\end{property}

The following properties are to hold in a fixed open interval $I$.
\begin{property}[SLI]\label{SLI}
Denote by $\chi_{x,L}$ the characteristic function of $\Lambda_L(x)$ and 
$\chi_x := \chi_{x,1}$. We also denote $\Gamma_{x,L}$ the characteristic function of $\Upsilon_L(x)$. Then for any compact interval $J\subset I$ there exists a finite constant $\gamma_J$ such that, given $L$, $l'$, $l''\in 2\NN$, $x$, $y$, $y'\in\ZZ^d$ with 
 $\Lambda_{l''}(y)\sqsubset\Lambda_{l'}(y')\sqsubset\Lambda_{L}(x)$, then for $\mathbb{P}$-almost every $\omega$, if $E\in J$ with $E\notin\sigma(H_{\omega,x,L })\cup\sigma(H_{\omega,y',l'})$
 we have
 \begin{equation}\label{eq:SLI-cond}
  \|\Gamma_{x,L}R_{\omega,x,L}(E)\chi_{y,l''}\|\leqslant\gamma_J\|\Gamma_{y',l'}R_{\omega,y',l'}(E)\chi_{y,l''}\|\|\Gamma_{x,L}R_{\omega,x,L}(E)\Gamma_{y',l'}\|.
 \end{equation}

\end{property}

\begin{property}[EDI]
 For any compact interval $J\subset I$ there exists a finite constant $\tilde{\gamma}_J$ such that for $\mathbb{P}$-almost every $\omega$, given a generalized eigenfunction $\psi$ of $H_\omega$
 with generalized eigenvalue $E\in J$, we have, for any $x\in\ZZ^d$ and $L\in2\NN$ with $E\notin\sigma(H_{\omega,x,L})$, that
 \begin{equation*}
  \|\chi_x\psi\|\leqslant\tilde{\gamma}_J\|\Gamma_{x,L}
  R_{\omega,x,L}(E)\chi_x\|\|\Gamma_{x,L}\psi\|.
 \end{equation*}

\end{property}

\begin{property}[NE]
 For any compact interval $J\subset I$ there exists a finite constant $C_J$ such that, for all $x\in\ZZ^d$ and $L\in2\NN$,
 \begin{equation*}
  \ec \left(\tr \left(E_{\omega,x,L}(J)\right)\right)\leqslant C_J L^d.
 \end{equation*}

\end{property}

\begin{property}[W]
 For some $b\geqslant 1$, there exists for each compact subinterval $J$ of $I$ a constant $Q_J$ such that
 \begin{equation}\label{eq:wegner}
  \mathbb{P}\{\dist (\sigma(H_{\omega,x,L}),E)<\eta\}\leqslant Q_J\eta L^{bd} ,
\end{equation}
for any $E\in J$, $0<\eta<\frac{1}{2}\dist(E_0,\sigma(H_0))$, $x\in\ZZ^d$ and $L\in2\NN$.
\end{property}

\begin{property}[H1($\theta$, $E_0$, $L_0$)]
 \begin{equation*}
  \pc\left\{\left\|\Gamma_{0,L_0}R_{\omega,0,L_0}(E_0)
  \chi_{0,L_0/3}\right\|
  \leqslant\frac{1}{L_0^\theta}\right\}>1-\frac{1}{841^d}.
 \end{equation*}
\end{property}

 \subsection{Multiscale analysis and localization}
 In this paragraph, we recall two very powerful results of Germinet and Klein which give us localization properties.
%
%
%
\begin{defi}
Given $E\in\RR$, $x\in\ZZ^d$ and $L\in6\NN$ with $E\notin \sigma(H_{\omega,x,L})$, we say that the box $\Lambda_L(x)$ is $(\omega,m,E)$-regular for a given $m>0$ if
\begin{equation}\label{eq:regular}
  \left\|\Gamma_{x,L}R_{\omega,x,L}(E)\chi_{x,L/3} \right\|
  \leqslant e^{-mL/2}.
\end{equation}
\end{defi}
%
%
%
In the following, we denote
 $$
 [L]_{6\NN}=\sup\{n\in6\NN | n\leqslant L\}.
 $$
\begin{defi}
For $x$, $y\in\ZZ^d$, $L\in6\NN$, $m>0$ and $I\subset\RR$ an interval, we denote.
\begin{equation*}
 R(m,L,I,x,y)=\left\{\omega;\text{for every }E'\in I\text{ either }\Lambda_L(x)\text{ or }\Lambda_L(y)\text{ is }(\omega,m,E')\text{-regular.}\right\}.
\end{equation*}
The multiscale analysis region $\Sigma_{MSA}$ for $H_\omega$ is the set of $E\in\Sigma$  for which there exists some open interval $I \ni E$
such that, given any $\zeta$, $0<\zeta<1$ and $\alpha$, $1<\alpha<\zeta^{-1}$, there is a length scale $L_0\in6\NN$ and a mass $m>0$ so if we set $L_{k+1}=[L_k^\alpha]_{6\NN}$, $k=0,1, \ldots $,
we have
\begin{equation*}
 \mathbb{P}\left\{R(m,L_k,I,x,y)\right\}\geqslant 1-e^{-L_k^\zeta}
\end{equation*}
for all $k\in\NN$, $x$, $y\in\ZZ^d$ with $|x-y|>L_k$.
\end{defi}
%
%
%
\begin{theorem}[Multiscale analysis - Theorem 5.4 p136 of \cite{klein}]\label{pi}
 Let $H_\omega$ be a Klein-standard ergodic random operator with (IAD) and properties (SLI), (NE) and (W) fulfilled in an open interval $I$. For $\Sigma$ being the almost sure spectrum of $H_\omega$ and for $b$ as in \eqref{eq:wegner}, given $\theta>bd$, for each $E\in I$ there exists a finite scale
 $\mathcal{L}_\theta(E)=\mathcal{L}_\theta(E,b,d)>0$ bounded on compact subintervals of $I$ such that, if for a  given $E_0\in\Sigma\cap I$ we have (H1)($\theta$, $E_0$, $L_0$) at some scale $L_0\in6\NN$ with $L_0>\mathcal{L}_\theta(E_0)$,
 then $E_0\in\Sigma_{MSA}$.
\end{theorem}
%
%
%
\begin{theorem}[Localization - Theorem~6.1 p139 of \cite{klein}]\label{thm:main-section4}
 Let $H_\omega$ be a Klein-standard ergodic operator with (IAD) and properties (SGEE) and (EDI) in an open interval $I$. Then,
 $$
 \Sigma_{MSA}\cap I\subset\Sigma_{EL}\cap\Sigma_{SSEHSKD}\cap I.
 $$
\end{theorem}
%
%
%
\section{Application to our setting}
We will now show that all the conditions listed in the previous Section hold true in our setting. 
\begin{theorem}\label{thm:main-in-the-text}
Let $H_\omega$ be the operator defined by \eqref{defop} obeying Assumptions~1-3. 
Then,  we have (IAD) and there exist two constants $E_\pm$ satisfying 
$B_-\leqslant E_-<\tilde{B}_-$ and $\tilde{B}_+<E_+\leqslant B_+$ such that 
(SLI), (EDI), (NE), (W), (SGEE) and (H1($\theta$, $\cdot$,$L_0$)) 
for $\theta$ and $L_0$ large enough are satisfied on $\Sigma\cap(E_-,E_+)$. 
Therefore, we have the localization properties (EL) and (SSEHSKD) on  
the interval $\Sigma\cap(E_-, E_+)$.
\end{theorem}
%
%
%
\begin{proof}
 (IAD) is a direct consequence of the independence of random variables stated in Assumption~\ref{assump:2} (i) and (ii).

To show (SLI), let $x$, $y$, $y'$, $L$, $l''$ and $l'$ be as in Property~\ref{SLI}  and consider, for $z\in\ZZ^d$ and $\ell>4$ a function 
$\tilde{\chi}_{z,\ell}\in C_0^\infty(\RR^d, [0,1])$ which has value 1 on $\Lambda_{\ell-3}(z)$ and 0 outside of $\Lambda_{\ell-5/2}(z)$ and whose gradient has norm smaller than 3.
Pick $E\in(B_-,B_+)$ such that $E\notin\sigma(H_{\omega,x,L })\cup\sigma(H_{\omega,y',l'})$.

Using Assumption~\ref{assump:2}(iii) on the support of $u$ leads us to the identity $H_\omega\tilde{\chi}_{y',l'}=H_{\omega,x,L}\tilde{\chi}_{y',l'}$ and then we get:
\begin{equation}\label{eq:pre-GRE}
(H_\omega-E)\tilde{\chi}_{y',l'}R_{\omega,x,L}(E)=\tilde{\chi}_{y',l'}
  +W_{y',l'}R_{\omega,x,L}(E)
\end{equation}
where \begin{equation*}
        W_{y',l'}=[H_\omega,\tilde{\chi}_{y',l'}] = [H_0, \tilde{\chi}_{y',l'}]
       \end{equation*}
is bounded according to Remark~\ref{rem:thm} (vi).

With similar support arguments, we have $H_\omega\tilde{\chi}_{y',l'}=H_{\omega, y',l'}\tilde{\chi}_{y',l'}$ and together with the identity~\eqref{eq:pre-GRE} we get the geometric resolvent equation:
\begin{equation}\label{star-star}
  \tilde{\chi}_{y',l'}R_{\omega,x,L}(E)=R_{\omega,y',l'}(E)\tilde{\chi}_{y',l'}
  +R_{\omega,y',l'}(E)W_{y',l'}R_{\omega,x,L}(E).
\end{equation}
Multiplying \eqref{star-star} from the left
by $\chi_{y,l''}$, from the right by $\Gamma_{x,L}$, writing $W_{y',l'}=\Gamma_{y',l'}W_{y',l'}\Gamma_{y',l'}$, $\tilde{\chi}_{y',l'} \Gamma_{x,L}=0$, and taking the norm of the adjoints, yields the estimate \eqref{eq:SLI-cond}.

For (EDI), we have, for $\psi$ a generalized eigenfunction of $H_\omega$ with associated generalized eigenvalue $E$:
\begin{equation*}
R_{\omega,x,L}(E)W_{x,L}\psi=R_{\omega,x,L}(E)\left(H_\omega\tilde{\chi}_{x,L}-\tilde{\chi}_{x,L}H_\omega\right)\psi.
\end{equation*}
But, denoting $V^{ext}_{\omega,x,L}=V_\omega-V_{\omega,x,L}$, we have,
\begin{equation*}
 H_\omega=H_{\omega,x,L}+V^{ext}_{\omega,x,L}=R_{\omega,x,L}(E)^{-1}+E+V^{ext}_{\omega,x,L}.
\end{equation*}
Then, 
\begin{equation*}
 R_{\omega,x,L}(E)W_{x,L}\psi=\tilde{\chi}_{x,L}\psi
 +R_{\omega,x,L}(E)E\tilde{\chi}_{x,L}\psi
 +R_{\omega,x,L}(E)V^{ext}_{\omega,x,L}\tilde{\chi}_{x,L}\psi
 -R_{\omega,x,L}(E) \tilde{\chi}_{x,L}H_\omega\psi.
\end{equation*}
Using the facts that $V^{ext}_{\omega,x,L}\,\tilde{\chi}_{x,L}=0$ and $H_\omega\psi=E\psi$, we get
 \begin{equation*}
  R_{\omega,x,L}(E)W_{x,L}\psi=\tilde{\chi}_{x,L}\psi
 \end{equation*}
which, through operations similar to the ones of the proof of (SLI), will give the desired result.

(NE) and (W) will be proved in Paragraph~\ref{Wegner}. (H1($\theta,E_0,L_0$)) for good values of the parameters will be proved in Paragraph~\ref{H1}.

Let us now give the proof of (SGEE). For the first part, we see that $\mathcal{D}^\omega_+\supset\mathcal{C}^\infty_c(\RR^d,\CC^n)$ which is dense in $\Hilb_+$ and a core for $H_\omega$ for any $\omega$.

For the second part we pick $T$ as in Section~\ref{S:section3.2}, being defined by the multiplication with $\langle x\rangle^{2\nu}$ where  $\nu>d/4$. 
Then we will show that for some $\lambda\in\RR$:
\begin{equation*}
  \tr\left(T^{-1}(H_\omega-i\lambda)^{-d}(H_\omega+i\lambda)^{-d}T^{-1}\right)
  \leqslant C,
\end{equation*}
with $C$ almost surely independent of $\omega$, 
which will imply (SGEE) for any interval  $I\subset \RR$, with $f:x\mapsto|x-i\lambda|^{-2d}$. 

To this purpose, it suffices to show that $T^{-1}(H_\omega-i\lambda)^{-d}$ is Hilbert-Schmidt with a  Hilbert-Schmidt norm almost surely independent of $\omega$.

For some $\alpha>0$, let $h_\alpha=\langle \cdot \rangle ^{\alpha}H_\omega\langle \cdot \rangle ^{-\alpha}$ defined on $\mathcal{C}^\infty_c(\RR^d,\CC^n)$. By using the fact that the multiplication by $\langle x\rangle ^{\pm\alpha}$ commutes with potentials, we find that for any $\phi\in \mathcal{C}^\infty_c(\RR^d,\CC^n)$
 $$
 h_\alpha\phi=H_\omega\phi+K\phi
 $$ 
for some bounded operator $K$ independent of $\omega$. 
We can then extend $h_\alpha$ on $\Dom(H_\omega)$.

Then, for $\lambda\in\RR^*$,
\begin{equation*}
 h_\alpha-\iu \lambda=\left(1+(W_\omega+K)(D_S-\iu \lambda)^{-1}\right)(D_S-\iu \lambda)
\end{equation*}
where $W_\omega=V_0+V_\omega$. As $(W_\omega+K)$ is bounded independently of $\omega$ and $\lambda$, we see that for $\lambda$ large enough $\|(D_S-\iu \lambda)^{-1}(W_\omega+K)\|<1$ so $h_\alpha-\iu\lambda$ is invertible.
Moreover, 
\begin{equation}
 (h_\alpha-\iu \lambda)^{-1}=(D_S-\iu \lambda)^{-1}\left(1+(W_\omega+K)
 (D_S-\iu  \lambda)^{-1}\right)^{-1}\label{born}.
\end{equation}
By a standard argument one can prove that the following identity holds:
$$\langle \cdot\rangle^{-\alpha}(h_\alpha-\iu \lambda)^{-1}= (H_\omega-\iu\lambda)^{-1}\langle \cdot\rangle^{-\alpha},$$
which together with \eqref{born} implies that: 
\begin{align}\label{hc10}
\langle \cdot\rangle^{\alpha}(H_\omega-\iu \lambda)^{-1}\langle \cdot\rangle^{-\alpha}=(D_S-\iu \lambda)^{-1}\left(1+(W_\omega+K)(D_S-\iu \lambda)^{-1}\right)^{-1}.
\end{align}
The idea is to write the operator $(H_\omega-\iu \lambda)^{-d}T^{-1}$ as a product of $d$ factors, each of them  belonging to $\mathcal{T}_{2d}$. In order to simplify notation, let us denote  $(H_\omega-\iu \lambda)^{-1}$ by $r$ and $T^{-1/d}$ with $t^{-1}$. Then we get by induction: 
\begin{align}\label{eq:star-star-star}
  (H_\omega-\iu \lambda)^{-d}T^{-1}
  &=r^d t^{-d}= r^{d-1} t^{-(d-1)} \{t^{-1} t^{d} rt^{-d}\} \nonumber \\
  &=\prod_{j=1}^d t^{-1} t^{j}rt^{-j}.
\end{align}
For each $j$, we can put  $\alpha= 2\nu j/d$ and by \eqref{hc10} we get:
$$t^{-1} t^jrt^{-j}=\langle \cdot\rangle^{-2\nu/d}(D_S-\iu \lambda)^{-1}\times {U_j}, $$
where $U_j$ is a bounded operator with a norm independent of $\omega$.
The function $\langle x\rangle^{-2\nu/d}$ belongs to $L^{2d}(\RR^d)$ when $\nu>d/4$.  Thus reasoning as in Remark~\ref{rem:thm}(iii) we have that 
 $(H_\omega-\iu \lambda)^{-d}T^{-1}$ is Hilbert-Schmidt with a norm which is independent of $\omega$. 
This proves $\mathrm{(SGEE)}$ and thus concludes the proof of Theorem~\ref{thm:main-in-the-text}.
\end{proof}

\subsection{Proof of (W) and (NE)}\label{Wegner}

Let $x\in\ZZ^d,L\in2\NN,\Lambda=\Lambda_L(x)$. We denote $\tilde{\Lambda}=\Lambda\cap\mathbb{Z}^d$. In order to alleviate notations, we denote $H_{\omega,\Lambda}=H_{\omega,x,L}$, $V_{\omega,\Lambda}=V_{\omega,x,L}$ and $E_{\omega,\Lambda}=E_{\omega,x,L}$ the spectral projector. We prove in this paragraph properties (W) and (NE) for the operator 
$H_{\omega,x,L}$,
namely we establish the following theorem.
%
%
%
\begin{theorem}[Wegner estimate]\label{thm:wegner}
Suppose Assumptions~\ref{assump-gap} and ~\ref{assump:2}(i)-(iii) hold true, 
and, for  $E_0\in (B_-,B_+)$ and  $\eta<\frac12\dist(E_0,\sigma( H_0))$, we denote $I_{\eta}(E_0)=[E_0-\eta, E_0+\eta]$. For any compact subinterval $J$ of $(B_-,B_+)$, there exists a constant $C_J$ such that for all $E_0\in J$
 \begin{equation*}
  \ecL \left ( \tr(E_{\omega,\Lambda} (I_\eta(E_0))) \right )\leqslant C_{J}\; \eta\; |\Lambda|.
\end{equation*}
\end{theorem}
%
%
%
\begin{remark}
This estimate trivially implies (NE). By Chebishev's inequality, 
it also leads to (W) with $b=1$.
\end{remark}

The resolvent of $H_0$ in $z\in\rho(H_0)$ will be denoted $R_0(z)$.
Let us fix some $E_0 \in (B_-, B_+)$ and denote $R_0 := R_0(E_0)$. The following proposition holds true:
%
%
%
\begin{prop}\label{prop:trace-estimate}
Assume that $E_0$ belongs to a compact $I$ in the gap. Let us denote 
\begin{equation*}
  K_{\{i\}}=u_{i_1} R_0u_{i_2} R_0 \cdots u_{i_{q-1}} R_0^2u_{i_q} ,
\end{equation*}
given a $q$-tuple $\{i\}$ for $q$ being an even integer larger than $2d$. 
Under Assumptions 1 and 2 (iii) on $V_{\omega,x,L}$, there exists a constant $C>0$ such that for all $E_0\in I$ we have
 \begin{equation}\label{eq:2.4}
  \sum_{i_1, \ldots ,i_q\in\tilde{\Lambda}}\|K_{\{i\}}\|_1\leqslant C|\Lambda|.
 \end{equation}
\end{prop}
%
For the proof of this Proposition we need the following two Combes-Thomas like lemmas which are proved in Appendix~\ref{appendix:combes-thomas}. 

\begin{lemme}\label{BC}
Fix a compact interval $I\subset (B_-,B_+)$. 
There exist two constants $\alpha>0$ and $C<\infty$ such that, for all $E\in I$ and any pair of  bounded functions $\chi_1$ and $\chi_2$ with $\|\chi_i\|_\infty\leqslant 1$ for $i=1,2$ and $\chi_1$ compactly supported, such that the distance between their supports is $a\geqslant 0$, 
 we have:
\begin{equation}\label{CT}
\|\chi_1(H_0-E)^{-1}\chi_2\|\leqslant C\;|{\rm supp}(\chi_1)|\; e^{-\alpha a}.
\end{equation}
\end{lemme}
The second lemma is a similar estimate with trace norm:
\begin{lemme}\label{ct-trace}
Let $a_0>0$. With the same notation as in Lemma \ref{BC}, assume that $a\geqslant a_0$. Then the operator $\chi_1 (H_0-E)^{-1} \chi_2$ is trace class and 
furthermore, there exist two constants $D>0$ and $\alpha>0$ such that for all  $E\in I$ and all $\chi_1$, $\chi_2$ satisfying the hypotheses in Lemma \ref{BC} we have
  \begin{equation}\label{eq:ct-trace}
  \|\chi_1 (H_0-E)^{-1}\chi_2\|_1\leqslant D\;|{\rm supp}(\chi_1)|\;e^{-\alpha a}.
  \end{equation}
\end{lemme}
%
%
%
The proof of these two lemmas are given in Appendix~\ref{appendix:combes-thomas}.
%
%
%
\begin{proof}[Proof of Proposition~\ref{prop:trace-estimate}]
The inequality \eqref{eq:2.4} is also proved in \cite[Proposition~7.2]{BCH} 
for Schr\" odinger operators under the assumptions that \eqref{CT} and \eqref{eq:ct-trace} hold true, although the authors do not consider moving centers $\xi_i(\omega)$. 

We omit here details of the proof since it is a straightforward adaptation of the proof of \cite[Proposition~7.2]{BCH} once Lemma~\ref{BC} 
and Lemma~\ref{ct-trace} are given.

The main ingredient behind the proof is that $u$ has compact support, thus keeping one index fixed, say $i_1$, the operator $K_{\{i\}}$ is trace class and 
$\sum_{i_2, \ldots ,i_q\in\tilde{\Lambda}}\|K_{\{i\}}\|_1$ is bounded by a numerical constant, uniformly on compacts in the gap. Note that if any two consecutive $u_{i_j}$ and $u_{i_{j+1}}$ have overlapping supports then we use that $u_{i_j} R_0\in \mathcal{T}_{2d}$, otherwise we use \eqref{eq:ct-trace} and control the series through the exponential localization. In the end we use that the number of terms $u_{i_1}$ is proportional with the Lebesgue measure of $\Lambda$.  
\end{proof}
%
%
%
For the proof of Theorem~\ref{thm:wegner}, we will use the following spectral averaging result proven in \cite[Corollary~4.2]{CH}.
%
%
%
\begin{prop}
\label{specav}
Let $H(\lambda)=H_0 + \lambda V$ a family of self-adjoint operators 
on a Hilbert space $\mathcal{H}$ where $V$ is bounded and satisfies
\begin{equation*}
 0\leqslant c_0B^2\leqslant V
\end{equation*}
for some $c_0>0$ and some bounded, self-adjoint operator $B$. Let $E_\lambda$ be the spectral family for $H(\lambda)$.
 Then, for any borelian $J\subset\mathbb{R}$ and any function $h\in L^\infty$ compactly supported, $h\geqslant 0$, 
 \begin{equation*}
\left\|  \int_\mathbb{R}h(\lambda)BE_\lambda(J)Bd\lambda\right\|\leqslant c_0^{-1}\|h\|_\infty|J|.
 \end{equation*}
 \end{prop}
%
%
%
%
%
%

\begin{proof}[Proof of Theorem~\ref{thm:wegner}]
The proof is very similar to the one in \cite{BCH} 
though it requires few technical changes. 
For the sake of completeness, we give it here.
 
 Let $J$ be a compact subinterval of $(B_-,B_+)$. 
We recall that if 
$H_{\omega,\Lambda}\psi_E = E\, \psi_E$, $E\in I_\eta(E_0)$, we have
\begin{align*}
K_0(E_0)\psi_E=-\psi_E + R_0(E_0)\, (H_{\omega,\Lambda} -E_0)\, \psi_E\ ,
\end{align*}
where $K_0(E_0) := R_0(E_0) V_{\omega,\Lambda}$. When there is no ambiguity, we will drop the dependence in $E_0$ in the notations.
Henceforth, 
\begin{align}
E_{\omega,\Lambda}(I_\eta)
= -K_0 E_{\omega,\Lambda}(I_\eta) + R_0 (H_{\omega,\Lambda}-E_0)
E_{\omega,\Lambda}(I_\eta)\ .\label{projspec}
\end{align}
Thus, noting that $E_{\omega,\Lambda}(I_\eta)$ is a positive
trace class operator,
\begin{align*}
\tr\left(E_{\omega,\Lambda}(I_\eta)\right) & 
=  \left\| E_{\omega,\Lambda}(I_\eta) \right\|_1  \\
 & \leqslant  \left| \tr(K_0 E_{\omega,\Lambda}(I_\eta))\right| 
 + \eta \left\| R_0 \right\| \, 
 \left\| E_{\omega,\Lambda}(I_\eta) \right\|_1 \ ,
\end{align*}
and since $\eta\leqslant\frac12 \mathrm{dist}(E_0, \sigma(H_0))$, we get
\begin{align}
\tr(E_{\omega,\Lambda}(I_\eta)) \leqslant 2\, |\tr(K_0E_{\omega,\Lambda}(I_\eta)) |.\label{eq:majo1}
\end{align}
A first consequence of \eqref{eq:majo1}  is, by the H\"older inequality with $q$ as in Proposition~\ref{prop:trace-estimate} and  $1/p + 1/q =1$,
\begin{equation}\label{eq:majo1-1}
\begin{split}
 \ecL \left(\| E_{\omega,\Lambda} (I_\eta) \|_1 \right)
 & \leqslant  2\, \ecL \left(\| K_0
 E_{\omega,\Lambda}(I_\eta) \|_1\right) 
 \leqslant  2\,  \ecL \left( \| K_0 \|_q \| E_{\omega,\Lambda}(I_\eta) \|_p \right)
 \\
 & \leqslant  2\,  \left\{ \ecL (\| K_0 \|_q ^q) \}^{1/q}
 \; \{ \ecL(\| E_{\omega,\Lambda}(I_\eta) \|_p^p) \right\}^{1/p} ,
 \end{split}
\end{equation}
where $\|\cdot\|_q$ denotes the norm in the Schatten class $\mathcal{T}_q$.

Since $q\geqslant 2d$, according to 
\eqref{schatten} we obtain  that there exists a constant $C$ such that for all $E_0\in J$ we have 
\begin{equation}\label{7.4}
\|K_0(E_0)\|_q\leqslant C  \|V_{\omega,\Lambda}\|_{L^q} 
 \leqslant C M_\infty |\Lambda|^{1/q}
\end{equation}
where  $M_\infty$ is defined by \eqref{Minfini}.

From this inequality, the
fact that $\ecL(\| E_{\omega,\Lambda}(I_\eta) \|_p^p) =
\ecL(\| E_{\omega,\Lambda}(I_\eta) \|_1 )$ (a consequence of the fact that 
the non-zero eigenvalues of the spectral projector are equal to one) and \eqref{eq:majo1-1}, we obtain:
\begin{equation}
 \ecL (\| E_{\omega,\Lambda}(I_\eta(E_0))\|_1 )
 \leqslant C\, | \Lambda | ,\label{(4.5)}
\end{equation} 
for all $E_0\in J$ which in particular ends the proof of Property~(NE).

Now, we use the adjoint of formula \eqref{projspec} to derive
 $$ 
 K_0 E_{\omega,\Lambda}(I_\eta) = -K_0 E_{\omega,\Lambda}(I_\eta) K_0^\ast 
 + K_0 E_{\omega,\Lambda}(I_\eta) (H_{\omega,\Lambda} -E_0) R_0 ,$$
which implies
\begin{equation}\label{(4.8.0)}
\left| \tr(K_0 E_{\omega,\Lambda}(I_\eta)) \right| 
 \leqslant  \left\| K_0 E_{\omega,\Lambda}(I_\eta) \right\|_1 
 \leqslant  \tr \left( K_0 E_{\omega,\Lambda}(I_\eta) K_0^\ast \right) 
 + \eta \left\| R_0 \right\| \, 
 \left\| K_0 E_{\omega,\Lambda}(I_\eta) \right\|_1 .
\end{equation}
Hence, by \eqref{eq:majo1} and $\eta \leqslant \frac12 \mathrm{dist}(E_0,\sigma(H_0))$, this yields
\begin{align*}
 \ecL \left( \tr (E_{\omega,\Lambda}(I_\eta) \right) 
 \leqslant 4 \, \ecL \left(\tr (K_0 E_{\omega,\Lambda}(I_\eta)
 K_0^\ast) \right) .
\end{align*}
If $q>2$, one continues this procedure and writes:
\begin{align}\label{(4.8.1)}
K_0 E_{\omega,\Lambda}(I_\eta) K_0^\ast = - K_0 E_{\omega,\Lambda}(I_\eta) (K_0^\ast)^2
+ K_0 E_{\omega,\Lambda}(I_\eta) (H_{\omega,\Lambda} - E_0) R_0 K_0^\ast .
\end{align}
One has by H\"older's inequality,
\begin{equation}\label{(4.8)}
\begin{split}
| \tr (K_0 E_{\omega,\Lambda}(I_\eta) (H_{\omega,\Lambda}-E_0) R_0 K_0^\ast) | 
 &  \leqslant
 \| K_0 E_{\omega,\Lambda}(I_\eta) (H_{\omega,\Lambda} -E_0) R_0 K_0^\ast \|_1  \\
 & \leqslant   \eta \| R_0 \|
\| K_0 E_{\omega,\Lambda}(I_\eta) \|_{q / ( q-1 ) } \| K_0^\ast\|_q  \\ 
 & \leqslant 
 \eta \| R_0 \| \| K_0 \|_{ q }^2
\| E_{\omega,\Lambda} (I_\eta) \|_{ q / ( q - 2 ) }  .
\end{split}
\end{equation}
Taking the expectation and again using H\"older's inequality, inequality \eqref{7.4} and
\eqref{(4.5)}, one can bound the expectation of the left hand side of \eqref{(4.8)} by
 $C  \eta  |\Lambda|$, where $C$ is a constant independent of $\eta$, $|\Lambda|$ and $E_0\in J$.  
Consequently, the latter equations \eqref{(4.8.0)}-\eqref{(4.8)} imply
 $$
 \ecL (\tr ( E_{\omega,\Lambda}(I_\eta) ))
 \leqslant  4\, \ecL
 (|\tr (K_0 E_{\omega,\Lambda}(I_\eta)
 (K_0^\ast)^{2}) |) + C  \eta| \Lambda | . 
 $$
If $q>3$, one repeats this procedure again. Finally, one obtains
\begin{equation}\label{(4.9)}
 \ecL \left(
 \tr ( E_{\omega,\Lambda}(I_\eta) \right)
 \leqslant 4\, \ecL
 \left( |\tr (K_0 E_{\omega,\Lambda}(I_\eta)
 (K_0^\ast)^{q-1} )| \right) 
 + C\, \eta | \Lambda | ,
\end{equation}
where $C$ is independent of $\eta$, $|\Lambda|$ and $E_0\in J$.

To estimate the first term on the right hand side of \eqref{(4.9)}, we expand the
potential $V_{\Lambda}= \sum_{i\in\widetilde{\Lambda}} \lambda_i u_i(\cdot - \xi_i)$.
In the rest of this proof, by abuse of notation, we shall denote $u_i(\cdot - \xi_i)$ by $u_i$. Moreover, we fix the values of all $\xi_i$'s, and expectation will be taken only  with respect to the $\lambda_i$'s. For each $q$-tuple of indices 
$\{i\}:= (i_1, \ldots , i_{q})
\in\widetilde{\Lambda}^{q}$, we define:
$$
K_{\{i\}} := K_{i_1\ldots  i_{q}}:= u_{i_2}^{\frac{1}{2}} R_0 u_{i_3}
R_0 u_{i_4}\cdots u_{i_{q}} R_0^2 u_{i_1}^{\frac{1}{2}}.
$$
By using H\"older's inequality for trace ideals \cite[Theorem~2.8]{simon},  

$K_{i_1\ldots i_{q}} \in \mathcal{T}_1$. In
terms of this operator, using cyclicity of trace, the first term on the right side of \eqref{(4.9)} becomes 
\be\label{(4.11)}
 \ecL
 \left( |\tr (K_0 E_{\omega,\Lambda}(I_\eta)
 (K_0^\ast)^{q-1} )| \right) 
 =\ecL \left\{\sum_{i_1,
 \ldots i_{q}\in\widetilde{\Lambda}}\lambda_{i_1}
 (\omega)\cdots\lambda_{i_{q}}(\omega) \tr\left\{
 K_{\{i\}} (
 u_{i_1}^{\frac{1}{2}} E_{\omega,\Lambda} (I_\eta)
 u_{i_2}^{\frac{1}{2}})\right\}\right\}.
\ee
Since $K_{\{i\}}$ is compact, we write it in terms of its singular
value decomposition. For each multi-index $\{i\}$, there
exists a pair of orthonormal bases,
$\left\{\phi_k^{\{i\}}\right\}$ and
$\left\{\psi_k^{\{i\}}\right\}$, and non-negative numbers
$\left\{\mu_k^{\{i\}}\right\}$, all independent of $\omega$, such
that
\be\label{(4.12)}
K_{\{i\}}= \sum_{k=1}^{\infty}
\mu_k^{\{i\}}\left\vert\phi_k^{\{i\}}
\rangle \langle \psi_k^{\{i\}}\right\vert .
\ee
Inserting the representation \eqref{(4.12)} into \eqref{(4.11)} and expanding the
trace in $\{\phi_k^{\{i\}}\}$, we obtain
\be\label{eq:k-sum1}
 \ecL
 \left\{\sum_{\{i\}\in\widetilde{\Lambda}^q}\sum_{k\geqslant 1}
 \lambda_{\{i\}}(\omega)\mu_k^{\{i\}}\langle \psi_k^{\{i\}},
 (u_{i_1}^{\frac{1}{2}}
 E_{\omega,\Lambda} (I_\eta) u_{i_2}^{\frac{1}{2}})\phi_k^{\{i\}}
 \rangle\right\},
\ee
where $\lambda_{\{i\}}(\omega):=
\lambda_{i_1}(\omega) \cdots \lambda_{i_{q}}(\omega)$. Recalling  that
$E_{\omega,\Lambda} (I_\eta) \geqslant 0 $, we bound the $k$-sum in \eqref{eq:k-sum1}
by
\be\label{(4.14)}
\begin{array}{r}
 \frac{1}{2} \displaystyle{ \sum_{k\geqslant 1} }^{}\mu_k^{\{i\}}
 \ecL \left\{
 |\lambda_{\{i\}}(\omega)|\langle\psi_k^{\{i\}},
 (u_{i_1}^{\frac{1}{2}}
 E_{\omega,\Lambda} (I_\eta) u_{i_1}^{\frac{1}{2}})
 \psi_k^{\{i\}} \rangle\right.\\
 \left.
 +|\lambda_{\{i\}}(\omega)|\langle\phi_k^{\{i\}},
 (u_{i_2}^{\frac{1}{2}}
 E_{\omega,\Lambda} (I_\eta) u_{i_2}^{\frac{1}{2}})
 \phi_k^{\{i\}}\rangle\right\} .\\
\end{array}
\ee
From the independence of the $\lambda_i$'s, the spectral averaging result
(Proposition \ref{specav}) applied to each term
in \eqref{(4.14)} gives for the first term:
\be\label{(4.15)}
 \ecL \left\{|\lambda_{\{i\}}(\omega)|\langle\psi_k^{\{i\}},
 (u_{i_1}^{\frac{1}{2}}
 E_{\omega,\Lambda} (I_\eta) u_{i_1}^{\frac{1}{2}})
 \psi_k^{\{i\}}\rangle \right\}\leqslant C_1 \, \eta .
\ee
where $C_1$ is finite, independent of $k$, and independent of $E_0$ according to Assumption~2(i). 
From inequalities \eqref{(4.11)}, \eqref{(4.14)} and  \eqref{(4.15)}, we obtain as upper bound for the first term on the right
hand side of \eqref{(4.9)}: 
\be\label{(4.16)}
 \ecL \left(
 \tr ( E_{\omega,\Lambda}(I_\eta) \right) \leqslant
 C_1 \, \eta \sum_{i_1, \ldots ,
 i_{q}\in\widetilde{\Lambda}}
 \left(\Vert K_{\{i\}}\Vert _1\right)\ .
\ee
Applying Proposition~\ref{prop:trace-estimate} we can bound the above series by a constant times the Lebesgue measure of $\Lambda$, and this ends the proof of the Wegner estimate and of the theorem.
\end{proof}
%
%
%
\begin{remark}
In order to apply Theorem~\ref{pi} (\cite[Theorem~5.4, p136]{klein}) for proving Theorems~\ref{thm:spec-loc}, \ref{thm:dyn-loc} and \ref{thm:main-in-the-text}, it would be enough to have a Wegner-like estimate with $|\Lambda|$ raised to some high power. Thus we could have shown directly  using  \eqref{schatten} and H\"olders's inequality for trace ideals   that 
 $$
 \sum_{i_1, \ldots ,i_q\in\tilde{\Lambda}}\|K_{\{i\}}\|_1 \leqslant C|\Lambda|^q.
 $$
In this way we would have avoided the use of Proposition~\ref{prop:trace-estimate}. 
\end{remark}

\subsection{Proof of (H1($\theta$,$E_0$,$L_0$))}\label{H1}
In this subsection, we want to prove
 \begin{equation*}
  \pc\left\{\|\Gamma_{0,L_0}R_{\omega,0,L_0}(E_0)\chi_{0,L_0/3}\|
  \leqslant\frac{1}{L_0^\theta}\right\}>1-\frac{1}{841^d}
 \end{equation*}
 for $E_0$ close enough to band edges $\tilde{B}_\pm$, some $\theta > d$ and $L_0$ large enough.
As  in \cite{BCH}, we first prove that, for $\delta>0$ small, $\dist (\sigma(H_{\omega,x,L}),\tilde{B}_\pm)>\delta$ with good probability. 
We can then apply Lemma~\ref{lemma:app-1} to get exponential decay of the resolvent at energies $E\in(\tilde{B}_--\delta/2,\tilde{B}_-]\cup[\tilde{B}_+,\tilde{B}_++\delta/2)$.
We finally verify $ H1(\theta,E_0, L_0)$ for any $\theta>0$, $E_0\in(\tilde{B}_--\delta/2,\tilde{B}_-]\cup[\tilde{B}_+,\tilde{B}_++\delta/2)$ and $L_0>L_0^*$ for some $L_0^*$
depending only on $\theta$, $d$, $B_\pm$ $\tilde{B}_\pm$, $\delta$, $M$, $m$ and $M_\infty$.

As in the previous section, we define $\Lambda=\Lambda_L(0)$ for some $L\in 2\NN$. We denote $\tilde{\Lambda}=\Lambda\cap\mathbb{Z}^d$, $H_{\omega,\Lambda}=H_{\omega,0,L}$, $V_{\omega,\Lambda}=V_{\omega,0,L}$.
\begin{lemme}\label{vp}
Let $\mu=\mu_{\omega_0,\Lambda}\in\sigma(H_{\omega_0,\Lambda})\cap (B_-,B_+)$ for some $\omega_0\in\Omega$. Then $\mu\in\Sigma$.
\end{lemme}

\begin{proof}

It is \eqref{spectre}. See also \cite[Lemma~5.1]{BCH} for an alternative proof that 
can easily be adapted for first order  operators.
\end{proof}


\begin{prop}
 Let $\delta_\pm=\frac{1}{2}|\tilde{B}_\pm-B_\pm|$ and $0<\delta<\frac{1}{2}M_\infty^{-1}\min(\delta_+,\delta_-)^2$. Assume that 
 $$
 \forall i\in\tilde{\Lambda},\ -(1-\delta M_\infty\min(\delta_+,\delta_-)^{-2})m
 <\lambda_i(\omega)<(1-\delta M_\infty\min(\delta_+,\delta_-)^{-2})M.
 $$
Then we have 
 $$
 \sup\left\{\sigma(H_{\omega,\Lambda})\cap(-\infty,\tilde{B}_-) \right\}
 <\tilde{B}_- -\delta
 $$
and
 $$
 \inf\left\{\sigma(H_{\omega,\Lambda})\cap(\tilde{B}_+,+\infty) \right\}
 >\tilde{B}_+ +\delta.
 $$
\end{prop}
%
%
\begin{proof}

We only prove the first inequality, the proof of the second one is similar.

Assume that the statement is false, i.e. there exist some $\Lambda$ and some values of the parameters $\lambda_i(\omega)$ and $\xi_i(\omega)$ such that $H_{\omega,\Lambda}$ has an
eigenvalue $\mu \in [\tilde{B}_- - \delta, \tilde{B}_-]$. If one of the coupling constants $\lambda_i$ is  negative, say $\lambda_0<0$, then  let us consider the family
 $$
  H(\lambda) := D_S +\lambda u(\cdot -\xi_0(\omega)) 
  +\sum_{i\neq 0, i\in \tilde{\Lambda}} \lambda_i(\omega) u(\cdot -\xi_i(\omega)-i),
  \quad \lambda \in [\lambda_0(\omega),0].
 $$
We have that $H(\lambda)$ is a self-adjoint analytic family of type (A) (cf. \cite[VII,\S 2]{kato}) and all its discrete eigenvalues $E_n(\lambda)$ in the interval $[\tilde{B}_- - \delta, \tilde{B}_-]$ can be followed real-analytically as functions of $\lambda$. Also, we may construct real analytic eigenvectors $\psi_n(\lambda)$ for each of them.  The Feynman-Hellmann formula and Assumption \ref{assump:2}(iii)  give:
 $$
 E_n'(\lambda)
 =\langle \psi_n(\lambda),u(\cdot-\xi_0(\omega))\psi_n(\lambda)\rangle 
 \geqslant 0,
 $$
which shows that $H(\lambda)$ will continue to have eigenvalues in $[\tilde{B}_- - \delta, \tilde{B}_-]$ up to $\lambda=0$. By induction, we may replace all the negative $\lambda_i$'s with zero, not changing the fact that the new realisation of $H_\omega$, this time with $V_{\omega,\Lambda}\geqslant 0$, still has at least one eigenvalue $\mu \in [\tilde{B}_- - \delta, \tilde{B}_-]$.

Now let us also assume that $V_{\omega,\Lambda}\geqslant 0$ and consider the analytic family of type (A) $T(\vt) := H_0+\vt V_{\omega,\Lambda}$,
for $\vt$ in a small real neighbourhood of $\vt_0=1$. 
Since $\mu$ has finite multiplicity, say $n$, there are at most $n$ functions
$\mu^{(k)}(\vt)$ analytic in $\vt$ near
$\vt_0=1$ such that
$\mu^{(k)}(1) =\mu$. Let
$\phi^{(k)}(\vt)$ be a real analytic eigenfunction for $\mu^{(k)}(\vt)$,
with $\Vert \phi^{(k)} (\vt)\Vert = 1$ for $\vt$ real and
$\vert\vt-1\vert$ small.
Applying the Feynman-Hellmann formula we find that for $\vt$ such that $\vt V_{\omega,\Lambda}\leqslant M_\infty$
\begin{equation}\label{eq:FH2}
\begin{split}
\frac{d\mu^{(k)}(\vt)}{d\vt}
& =\langle\phi^{(k)}(\vt), V_{\omega,\Lambda}\phi^{(k)}(\vt)\rangle
\geqslant \vt^{-1} M_\infty^{-1} \| \vt
V_{\omega,\Lambda}\phi^{(k)} (\vt)\|^2 \\
& = \vt^{-1} M_\infty^{-1} 
\left\|
\left(H_0 - \mu^{(k)}_\vt \right) \phi^{(k)} (\vt) \right\|^2 
\geqslant \vt^{-1} M_\infty^{-1}  \left(\dist (\sigma(H_0), \mu^{(k)}_\vt)\right)^2.
\end{split}
\end{equation}
We now assume $\lambda_i(\omega)<(1-\delta M_{\infty} [\min
(\delta_+, \delta_-)]^{-2}) M, \forall i \in\widetilde{\Lambda}$,
and fix

\begin{align}\label{hoc1}
\vt_1 = {\min}_{i \in\widetilde{\Lambda}}^{} \left(\frac
{M}{\lambda_i(\omega)}\right) \geqslant \left(1-\delta M_{\infty}
\left[\min(\delta_+,
\delta_-)\right]^{-2}\right)^{-1}>1 .
\end{align}
We see  that by definition of $\vt_1$ the condition $\vt V_{\omega,\Lambda}
\leqslant M_\infty$ is satisfied on the interval $[1,\vt_1].$

Upon integrating \eqref{eq:FH2} over $[1,\vt_1]$ and using that $\mu\leqslant \mu^{(k)}(\vt)\leqslant \mu^{(k)}(\vt_1)$ we get:
\begin{equation}\nonumber
\mu^{(k)}(\vt_1) \geqslant
\mu+(\log\vt_1)M_{\infty}^{-1}\min\left\{\left[\dist (\mu^{(k)}(\vt_1),
\sigma(H_0))\right]^2,\, \left[\dist (\mu,
\sigma(H_0))\right]^2 \right\}.
\end{equation}

We have to bound the minimum of the distances. As we always have the following order 
  $$
  B_-<\mu\leqslant \mu^{(k)}(\vt_1)\leqslant\tilde{B}_-<\tilde{B}_+<B_+
  $$
there are only two cases:
\begin{itemize}
 \item either the minimum is $\dist (\mu^{(k)}(\vt_1),\sigma(H_0))$ and then it is equal to $B_--\mu^{(k)}(\vt_1)>2\delta_+$.
 \item or the minimum is $\dist (\mu,\sigma(H_0))$ and then it is equal to $\mu-B_-$. As $\mu>\tilde{B}_--\delta$, this distance is greater than $\tilde{B}_--\delta-B_-=2\delta_--\delta$.
 As $\delta<\frac{1}{2}M_\infty^{-1}\delta_-^2$, the distance is larger than $\delta_-(2-\frac{1}{2}M_\infty^{-1}\delta_-)$. Using Lemma~\ref{lemaunu} with $A-B=V_\omega$ and $\|V_\omega\|\leqslant M_\infty$, we must have $2\delta_-\leqslant M_ \infty$ so the distance is larger than $\frac{3}{2}\delta_-$.
\end{itemize}

Thus the minimum is larger than $\frac{3}{2}\min( \delta_+, \delta_-)$. Then using the inequality $-\log(1-x)\geqslant x$ with $x=1-\vt_1^{-1}$ from \eqref{hoc1} we have
$$\log(\vt_1)\geqslant 1-\vt_1^{-1}=\delta M_{\infty}
\left[\min(\delta_+,
\delta_-)\right]^{-2}$$
which leads to $ \mu^{(k)}(\vt_1)>\tilde{B}_-$ and thus to a contradiction. 
\end{proof}
%
%
\begin{coro}\label{cor:2.10}
For $0<\delta<\frac{1}{2}M_\infty^{-1}\min(\delta_+,\delta_-)$, we have 
 $$
 \sup\left(\sigma(H_{\omega,\Lambda})\cap(-\infty,\tilde{B}_-) \right)
 <\tilde{B}_- -\delta 
 $$ 
and
 $$
 \inf\left(\sigma(H_{\omega,\Lambda})\cap(\tilde{B}_+,+\infty) \right)
 >\tilde{B}_+ +\delta ,
 $$
with probability larger than
 $$
  1-2|\Lambda|\max_{X\in \{-m,M\}}
  \left|\int_{(1-\delta M_\infty\min(\delta_+,\delta_-)^{-2})X}^Xh(s)ds\right|.
 $$ 
\end{coro}
%
%
%
\begin{proof}
The probability that 
 $$
 \forall i\in\tilde{\Lambda},\ -(1-\delta M_\infty\min(\delta_+,\delta_-)^{-2})m
 <\lambda_i(\omega)<(1-\delta M_\infty\min(\delta_+,\delta_-)^{-2})M
 $$
is given by
 $$
 \left[1-\int_{(1-\delta M_{\infty}[\min(\delta_+,\delta_-)]^{-2}) M}^M h(s)ds
-\int_{-m}^{-(1-\delta M_{\infty}[\min(\delta_+,\delta_-)]^{-2}) m} h(s)ds
 \right]^{|\Lambda|} .
 $$
The conclusion follows by using $(1-x)^\alpha\geqslant 1-\alpha x$ for $\alpha>1$ and $x\in [0,1]$.
\end{proof}
%
%
%
We can now prove hypothesis $(H1(\theta,E_0,L_0))$.
\begin{prop}
 Let $\chi_i,i = 1,2$, be two functions with $\Vert
\chi_i\Vert_{\infty}\leqslant 1, \supp (\chi_1)\subset \Lambda_{L_0/3}$
and $\supp (\chi_2)\subset \Lambda_{L_0}$ such that $\sup_{x\in\supp (\chi_2)}\dist (x, \partial \Lambda_{L_0}) < L_0/8$. Define 
$\delta_{\pm}:=\frac{1}{2}| \widetilde{B}_\pm-B_\pm|$.
For $\beta>0$ as in Assumption \ref{assump:2} (iv), consider any $\nu>0$ such that $0<\nu<4 \beta (2\beta
+d)^{-1}<2$. Then there exists $L_0^*$ such that for all $L_0>L_0^*$ and $E_0\in 
 ]\tilde{B}_--L_0^{\nu-2},\tilde{B}_-]\cup[\tilde{B}_+,\tilde{B}_++L_0^{\nu-2}[$,
 $$
 \sup_{\epsilon>0} \| \chi_2R_{\Lambda_{L_0}}(E_0 +i\epsilon)\chi_1 \|
 \leqslant e^{-L_0^{\nu/3}} ,
 $$
 with probability larger than $1-\frac{1}{841^d}$.
\end{prop}
\begin{proof}
 Pick $\delta = 2L_0^{\nu-2}$. For $L_0$ large enough we have $\delta<\frac12 M_\infty^{-1}\min(\delta_+,\delta_-)^2$, hence, using Assumption~\ref{assump:2}(iv), Corollary~\ref{cor:2.10} and the fact that $0<\nu<4\beta(2\beta+d)^{-1}$ yields 
 \begin{equation*}
  \mathbb{P}\left\{\dist (\sigma(H_{\omega,0,L_0}),\tilde{B}_\pm)
  >\delta \right\}
  \geqslant 1-2 L_0^d\left( \max(m,M) \delta M_\infty\min(\delta_+,\delta_-)^{-2}  \right)^{d/2+\beta}
  \geqslant1-\frac{1}{841^d}, 
 \end{equation*}
for $L_0$ large enough.

Now consider any realisation of $H_{\omega,0,L_0}$ which obeys $\dist (\sigma(H_{\omega,0,L_0}),\tilde{B}_\pm)
  >\delta=2L_0^{\nu-2}$ and let $E_0\in 
 ]\tilde{B}_--L_0^{\nu-2},\tilde{B}_-]\cup[\tilde{B}_+,\tilde{B}_++L_0^{\nu-2}[$. 
We now apply Lemma~\ref{lemma:app-1} with $x_0=0$, knowing that, for $a_1$ and $a_2$ as defined in Lemma~\ref{lemma:app-1}, we have $a_2 - a_1 \geqslant L_0/8$. We get
$$
 \|\chi_2 R_{\Lambda_{L_0}}(E+i\epsilon)\chi_1\| 
 \leqslant\frac{2}{L_0^{\nu-2}}
 \exp\left(- c\frac{L_0^{\nu/2-1}}{2}|
 \tilde{B}_+ - \tilde{B}_-|^{1/2}\frac{L_0}{8}\right).
$$
The result follows by taking $L_0$ large enough.
\end{proof}

Property~$(H1(\theta,E_0,L_0))$ comes directly from the previous proposition as $\chi_{0,L_0/3}$ and $\Gamma_{0,L_0}$ satisfy its hypotheses 
and $e^{-L_0^{\nu/3}}\leqslant \frac{1}{L_0^\theta}$ when $L_0>\mathcal{L}_\theta$ for some finite $\mathcal{L}_\theta$.
%

\appendix

\section{Spectrum location}\label{specloc}

\subsection{Proof of Proposition \ref{pastout}}

\begin{lemme}\label{lemadoi}
Let $\tilde{u}:\RR^d\mapsto \mathcal{H}_n(\CC)$ be a bounded, compactly supported, 
non-negative matrix valued multiplication potential which is not identically zero.
Let $H_0$ be defined by \eqref{hc1} and define
\begin{equation*}
 H_\tau := H_0+\tau \tilde{u}(x),\quad \tau\in \RR.
\end{equation*}
Then there exists some $\tau\in \RR$ with $|\tau|>0$ such that $H_\tau$has at least one discrete eigenvalue in $(B_-,B_+)$. 
\end{lemme}
\begin{proof}
The perturbation given by $\tilde{u}$ is relatively compact to $H_0$, hence due to the Birman-Schwinger principle we have that $\mu \in (B_-,B_+)$ is a discrete eigenvalue of $H_\tau$ if $-1$ is an eigenvalue of $\tau  \tilde{u}^{1/2}(H_0-\mu)^{-1}\tilde{u}^{1/2}$.
The family of self-adjoint operators $T(\mu) := \tilde{u}^{1/2}(H_0-\mu)^{-1}\tilde{u}^{1/2}$ cannot be identically zero for $\mu\in (B_-,B_+)$ because this would lead to 
 $$
	T'(\mu)=\tilde{u}^{1/2}(H_0-\mu)^{-2}\tilde{u}^{1/2} \equiv 0 ,
 $$
hence $|H_0-\mu|^{-1} \tilde{u}^{1/2}=0$ and $\tilde{u}^{1/2}=0$, contradiction. Now let $\mu_0\in (B_-, B_+)$ be such that $T(\mu_0)$ has a non-zero real eigenvalue $E_0$. Then choosing $\tau_0=-1/E_0$ we obtain that $H_{\tau_0}$ has a discrete eigenvalue at $\mu_0$. 
\end{proof}

A slightly more general version of the following lemma can be found in \cite[V,Theorem~4.10]{kato}

The Hausdorff distance between two real subsets $\Omega_{1,2}\subset \RR$ is defined as 
\begin{align}\label{hausd}
 d_H(\Omega_1,\Omega_2) := 
 \max\{ \sup_{x\in \Omega_1}{\rm dist}(x,\Omega_2), \sup_{y\in \Omega_2}
 {\rm dist}(x, \Omega_1)\}.
\end{align} 

\begin{lemme}\label{lemaunu}
Let $A$ and $B$ be two self-adjoint operators acting on the same Hilbert space and having the same domain, such that $A-B$ is bounded. Then 
\begin{align}\label{LipZd.Lem-p2Spec}
d_H\big(\sigma(A),\sigma(B)\big) \;
	\leqslant\; \|A-B\|\,.
\end{align}
\end{lemme}
\begin{proof} Let $\lambda\not \in\sigma(A)$ such that $d(\lambda,\sigma(A))> \|A-B\|$. Then the operator $(B-A)(A-\lambda)^{-1}$ has norm less than $1$ and ${\rm Id} +(B-A)(A-\lambda)^{-1}$ is invertible with a bounded inverse. Thus
$$B-\lambda=\left ({\rm Id} +(B-A)(A-\lambda)^{-1}\right )(A-\lambda)$$
is also invertible with a bounded inverse, which shows that $\lambda\not\in \sigma(B)$. In other words, no element of $\sigma(B)$ can be located at a distance larger than $\|A-B\|$ from $\sigma(A)$, which implies:
$$\sup_{E\in \sigma(B)}d(E,\sigma(A))\leqslant \|A-B\|.$$
By interchanging $A$ with $B$, the proof is over. 
\end{proof}

\begin{lemme}\label{lematrei}
Using the notation and result of Lemma~\ref{lemadoi}, let $u := \tau_0 \tilde{u}$ and consider the operator $H_\omega$ as in \eqref{defop}. With the notation introduced in Assumption \ref{assump:2}(i), let $m, M\in (1,2)$. Then there exists $\lambda_0\in (0,1)$ small enough such that Assumption \ref{gap} is satisfied if $m$ and $M$ are replaced respectively by $\lambda_0m$ and $\lambda_0 M$. 
\end{lemme}
\begin{proof}
For the sake of simplicity, let us assume $m=M$. According to Lemma~\ref{lemadoi}, we know that some $\mu_0\in (B_-,B_+)$ belongs to the spectrum of $H_{\tau_0}=H_0+u(x)$. 
Using \eqref{spectre}, one can show that $\mu_0$ also belongs to the spectrum of $H_\omega$ for $\omega$ belonging to a set of measure one, hence $\mu_0$ belongs to the almost sure spectrum $\Sigma$.

Now consider the family $H_{\lambda,\omega}:= H_0+ \lambda V_\omega$ with $\lambda\in (0,1)$. By multiplying the potential with $\lambda$ we effectively reduce the support of $h$ to $[-M\lambda, M\lambda]$. Because $V_\omega$ is uniformly bounded for all $\omega$, we know from Lemma \ref{lemaunu} that the spectrum $\sigma(H_{\lambda,\omega})$ varies Lipschitz continuously with $\lambda$, uniformly in $\omega$. 

We now want to prove that the almost sure spectrum $\Sigma_\lambda$ is continuous in $\lambda$ in the Hausdorff distance. Let $E\in \Sigma_\lambda$ and fix $\epsilon>0$. There exists some $\omega_E$ such that $E\in \sigma(H_{\lambda\omega_E})$. By the Weyl criterion, there exists $\psi_E$ of norm one such that 
$$\|(H_{\lambda\omega_E}-E)\psi_E\|\leqslant \epsilon/10.$$ 
Then there exists some $\Lambda := \Lambda_{E,\epsilon,\lambda}\subset \RR^d$ large enough such that  $H_{\Lambda,\lambda\omega_E} := H_0+V_{\Lambda,\lambda\omega_E}$ obeys
$$\|(H_{\Lambda,\lambda\omega_E}-E)\psi_E\|\leqslant \epsilon/{5}.$$
This inequality implies by the same Weyl criterion that the operator $H_{\Lambda,\lambda\omega_E}$ must have at least one point $E'$ of its spectrum such that $E'\in (E-\epsilon/5,E+\epsilon/5)$. Now using Lemma \ref{lemaunu} we can find some $\delta>0$ such that for every $\lambda'$ obeying $|\lambda'-\lambda|<\delta$, the Hausdorff distance between the spectra of $H_{\Lambda,\lambda\omega_E}$ and $H_{\Lambda,\lambda'\omega_E}$ is less than $\epsilon/5$ thus there must exist $E''$ in $\sigma(H_{\Lambda,\lambda'\omega_E})$ such that $|E''-E|<\epsilon$. Finally, via Kirsch's argument \eqref{spectre} one can prove that $E''$ belongs to the almost sure spectrum of $H_{\lambda' \omega}$; in other words, 
$$\sup_{E\in \Sigma_{\lambda}}d(E,\Sigma_{\lambda'})<\epsilon,\quad \forall |\lambda'-\lambda|<\delta.$$

This implies in particular that the almost sure spectrum of $H_{\lambda,\omega}$ must converge (as a set) to the spectrum of $H_0$ when $\lambda$ tends to zero.  Thus if $\lambda$ is small enough, then at least one gap must appear in the almost sure spectrum of $H_{\lambda\omega}$, which due to the same continuity, it must still have some non-empty component in the old gap $(B_-,B_+)$. 
\end{proof}

\subsection{Proof of Proposition \ref{projspec2}}
Under the conditions of Lemma \ref{lematrei} we know that there exists a gap $[B_-',B_+']\subset (B_-,B_+)$ in the almost sure spectrum $\Sigma$ of $H_\omega$, and at the same time, either $\Sigma\cap (B_-,B_-')$ or 
$\Sigma \cap (B_+',B_+)$ is non-empty. 

Now assume that $\Sigma\cap (B_-,B_-')$ is not empty. Let $\tilde{B}_-\in (B_-,B_-')$ be the supremum of this set (note that $\tilde{B}_-<B_-'$ since $\Sigma$ is closed and we must have $\tilde{B}_-\in \Sigma$). If $\lambda\in [0,1]$ we consider the family $H_{\lambda\omega}$ and denote by $\Sigma_\lambda$ its almost sure spectrum. As a set, $\Sigma_\lambda$ varies continuously with $\lambda$ in the Hausdorff distance as we saw in Lemma \ref{lematrei}. Denote by $E_\lambda$ the supremum of $\Sigma_\lambda\cap  [B_-,B_-')$. Because $E_1=\tilde{B}_-$, $E_0=B_-$ and $E_\lambda$ varies continuously with $\lambda$, we conclude that $E_\lambda$ covers the interval $[B_-,\tilde{B}_-]$. Finally, since $E_\lambda\in \Sigma_\lambda\subset \Sigma$, we conclude that $[B_-,\tilde{B}_-]\subset \Sigma$, hence no other gaps can appear in this interval.

\section{Combes-Thomas estimates}

\label{appendix:combes-thomas}
This section is dedicated to Lemma~\ref{BC} and Lemma~\ref{ct-trace}.
The proof of Lemma~\ref{BC}  follows closely the strategy \cite[Proposition~5.2]{briet-cornean}.  
%
%
%
\begin{lemme}\label{lemma:app-1}
Let $W$ be a symmetric and matrix-valued bounded potential, and let $H = D_S +W$ where $D_S=S\sigma\cdot(-i\nabla)S$ is like in \eqref{hc1} and $S$ is a bounded coefficient operator as in \eqref{def:S}. Assume 
that $H$ has a gap $(E_-, E_+)$ in its spectrum, containing $0$. Consider $\chi_1$ and $\chi_2$ two  compactly supported 
functions such that $\|\chi_i\|_\infty \leqslant 1$. For $x_0\in\mathbb{R}^d$ define
$$
a_1=\sup_{x\in \mathrm{supp}(\chi_1)}|x-x_0| \quad\mbox{and}\quad a_2=\mathrm{dist}(x_0, \mathrm{supp}(\chi_2)).
$$
For $E\in (E_-, E_+)$ let
$$
 \upsilon_\pm = \mathrm{dist}(E, E_\pm) \quad\mbox{and}\quad \upsilon = \min(\upsilon_+, \upsilon_-).
$$
Then there exists a constant $c>0$ such that for all $E\in (E_-, E_+)$ we have:
\begin{equation}\label{BCapp}
 \displaystyle\| \chi_1 (H - E)^{-1}\chi_2\| \leqslant
 \displaystyle\frac{2}{\upsilon} e^{\displaystyle- c\sqrt{\upsilon_+\upsilon_-}(a_2-a_1)}.
\end{equation}
\end{lemme}

\begin{proof}
Let $\epsilon>0$ and define $\langle x-x_0\rangle_\epsilon := \sqrt{\epsilon+|x-x_0|^2}$. For $t>0$, we define on $\mathcal{C}^\infty_c (\mathbb{R}^d,\mathbb{C}^n)$ the (non self-adjoint) operator
 $$
 H_{t,\epsilon} := e^{-t\langle x-x_0\rangle_\epsilon }
 H e^{t\langle x-x_0\rangle_\epsilon } 
 =H-tS\sigma\cdot(\iu\nabla\langle x-x_0\rangle_\epsilon)S.
 $$ 
The operator $H_{t,\epsilon}$ is closed on the domain of $H$. Let 
$\psi \in \mathcal{C}^\infty_c (\mathbb{R}^d,\mathbb{C}^n)$ with norm $1$. 
We denote $\psi^-=P_{(-\infty,E_-]}\psi$ and  $\psi^+=P_{[E_+,+\infty)}\psi$, 
where $P$ are the spectral projectors for $H$, and we remind that 
$\upsilon_\pm=\mathrm{dist}(E,E_\pm)$.

 We have
 \begin{align*}
  \| (H_{t,\epsilon}-E)\psi \| \geqslant & 
  \Re(\langle \psi^+-\psi^-,(H_{t,\epsilon}-E)(\psi^++\psi^-)\rangle )\\
  \geqslant & \upsilon_+ \|\psi^+\|^2 + \upsilon_- \|\psi^-\|^2 
  -2 \|tS\sigma\cdot(\iu\nabla\langle x-x_0\rangle_\epsilon )S\| 
  \, \|\psi^+\| \, \|\psi^-\|.
  \end{align*}
We observe that the length of $\nabla\langle x-x_0\rangle_\epsilon$ is bounded by a number independent of $\epsilon$. Let $t := c\sqrt{\upsilon_+\upsilon_-}$ where $c>0$ is independent of both $E$ and $\epsilon$,  and small enough so that: 
 $$\|tS\sigma\cdot(\iu\nabla\langle x-x_0\rangle_\epsilon )S\| <\sqrt{\upsilon_+\upsilon_-}/2.$$ We then have
 $$
 \|(H_{t,\epsilon}-E)\psi\| \geqslant 1/2\min(\upsilon_+,\upsilon_-).
 $$
Thus, $H_{t,\epsilon}-E$ is invertible for $E\in (E_-,E_+)$ and 
 $$ 
 \|(H_{t,\epsilon}-E)^{-1} \| \leqslant\frac{2}{\upsilon},
 $$
uniformly in $\epsilon$. Hence, 
\begin{align*}
 \|\chi_1(H-E)^{-1}\chi_2\|= &\|\chi_1e^{t\langle \cdot-x_0\rangle_\epsilon }(H_{t,\epsilon}-E)^{-1}e^{-t\langle \cdot-x_0\rangle_\epsilon }\chi_2\|\\
 \leqslant &\|\chi_1 e^{t\langle \cdot -x_0\rangle_\epsilon }\|\,\|(H_{t,\epsilon}-E)^{-1}\|\,\|e^{-t\langle \cdot -x_0\rangle_\epsilon }\chi_2\|.
\end{align*}
The central factor is bounded by $2/\upsilon$. By taking $\epsilon$ to zero, the first factor is bounded by $e^{t a_1}$ and the third factor by $e^{-t a_2}$.
We have thus proved the lemma.
\end{proof}

\begin{proof}[Proof of Lemma \ref{BC}] 
Without loss of generality we may assume that the distance between the supports 
obeys $a\geqslant 10$. Let $K=[-1/2,1/2)^d$ be a unit cube in $\RR^d$. 
Let $g_\gamma$ be the characteristic function of the cube $K_\gamma := \gamma +K$, $\gamma\in \mathbb{Z}^d$. We have 
\begin{align}\label{hc21}
\chi_1(H_0-E)^{-1}\chi_2=\sum_{\gamma,\gamma'}g_\gamma\chi_1(H_0-E)^{-1}\chi_2g_{\gamma'}.
\end{align}

The sum over $\gamma$ only contains finitely many terms because $\chi_{1}$ is compactly supported. For any given such pair $g_\gamma\chi_1$ and $g_{\gamma'}\chi_2$ we apply Lemma \ref{lemma:app-1} in which we choose  $x_0=\gamma$. We observe that in this case $a_1\leqslant 1$ and since $a\geqslant 10$ we also have $a_2\geqslant a/3+|\gamma-\gamma'|/3$. Thus \eqref{BCapp} leads to
 $$
 \| g_\gamma\chi_1(H_0-E)^{-1}\chi_2g_{\gamma'}\|
 \leqslant c_1 e^{-c_2 a}e^{-c_2 |\gamma-\gamma'|}
 $$
where $C_1$ and $c_2$ are constants depending on the interval $I$. Then we can sum over $\gamma'$ for every fixed $\gamma$ and we are done. 
\end{proof}


We are ready to prove Lemma~\ref{ct-trace}. 
\begin{proof}[Proof of Lemma~\ref{ct-trace}]
Using the same notation as in the proof of Lemma \ref{BC}, the strategy is to show the existence of two positive constants $c_1$ and $c_2$ such that in the trace norm we have:
\begin{align}\label{hc20}
\| g_\gamma\chi_1(H_0-E)^{-1}\chi_2g_{\gamma'}\|_1
\leqslant c_1 e^{-c_2 a}e^{-c_2 |\gamma-\gamma'|}.
\end{align}
Without loss of generality  we may assume that $a_0=10$ and $a\geqslant 10$. Then the pairs $\gamma$ and $\gamma'$ which give a non-zero contribution must obey $|\gamma-\gamma'|\geqslant 8$. 

We now consider $2d$ smooth and compactly supported functions $0\leqslant f_j\leqslant 1$ which obey the following conditions: $g_\gamma f_1=g_\gamma$, $f_j f_{j+1}=f_j$ 
if $1\leqslant j\leqslant 2d$, and the support of the "largest" function $f_{2d}$ is contained in the hypercube centered at $\gamma$ with side-length $2$. In particular, the support of $f_j$ and the support of  the derivatives of $f_{j+1}$ are disjoint, and also $f_{2d}g_{\gamma'}=0$. 

Denote $R_0 := (H_0-E)^{-1}$. We have $[f_j ,R_0]=R_0 S(-\iu \sigma \cdot \nabla f_j)S R_0$ and 
$$g_\gamma R_0g_{\gamma'}=g_\gamma f_{2d} R_0g_{\gamma'} =g_\gamma R_0 S(-\iu \sigma \cdot \nabla f_{2d})S R_0 g_{\gamma'}$$
and repeating this for all $j$ we have:
$$g_\gamma R_0g_{\gamma'}=g_\gamma \prod_{j=1}^{2d} \left ( R_0 S(-\iu \sigma \cdot \nabla f_{j})S \right ) \chi_{{\rm supp}(f_{2d})}R_0 g_{\gamma'}.$$
Each factor $R_0 S(-\iu \sigma \cdot \nabla f_{j})S$ belongs to $\mathcal{T}_{2d}$ with a norm which is independent of $\gamma$ and $\gamma'$. Thus the product is trace class. Moreover, by applying Lemma \ref{lemma:app-1} to the pair $\chi_{{\rm supp}(f_{2d})}$ and $g_{\gamma'}$ with $x_0=\gamma$ we obtain $a_2-a_1\geqslant |\gamma-\gamma'|/10 +a/10$ and 
 $$
 \| \chi_{{\rm supp}(f_{2d})}R_0 g_{\gamma'}\|
 \leqslant Ce^{-\alpha a}e^{-\alpha |\gamma-\gamma'|}.
 $$
This proves \eqref{hc20}. Since there is a finite number of $g_\gamma$'s which give a non-zero contribution in \eqref{hc21}, this number being proportional with the Lebesgue measure of the support of $\chi_1$, the proof is over. 
\end{proof}

\subsection*{Acknowledgment}
J.-M.B. is grateful to P. M\"uller for sharing the reference \cite{POC2017} and for fruitful discussions. It is a pleasure to thank the REB program of CIRM for giving us the opportunity to start this research. H.C. acknowledges the support of the Simons-CRM Scholar-in-residence program during the preparation of this work.


\end{document}